\documentclass[webpdf,contemporary,large]{oup-authoring-template}%

\graphicspath{{Fig/}}

\usepackage{multirow}

\usepackage{wrapfig}

\usepackage{graphicx}
\usepackage[]{geometry}
\usepackage{xcolor}
\definecolor{darkgreen}{RGB}{34, 139, 34} 
\definecolor{ipcolor}{RGB}{242,170,60}
\definecolor{bpcolor}{RGB}{55,126,247}
\usepackage{xspace}
\usepackage{amssymb,amsmath,epsfig}
\usepackage{amsthm}
\usepackage{hyperref}  
\hypersetup{
    colorlinks=true,
    linkcolor=black, 
    urlcolor=blue,  
    citecolor=blue  
}

\usepackage{verbatim}
\algrenewcommand\algorithmicindent{0.5em} 

\renewcommand{\cite}[1]{\citep{{#1}}}  

\usepackage{hyperref}\hypersetup{colorlinks=true, citecolor=[rgb]{0,0,1}, urlcolor=black}

\usepackage{courier}
\usepackage{float}
\usepackage[skip=4pt]{caption}
\usepackage{subcaption}
\usepackage{tikz}
\usetikzlibrary{shapes,snakes}
\usepackage[capitalise]{cleveref}
\usepackage{tablefootnote}
\usepackage{tabularx}
\usepackage{seqsplit} 

\usetikzlibrary{shapes.geometric, arrows, positioning}

\tikzstyle{process} = [rectangle, rounded corners, minimum width=2.0cm, minimum height=1cm, text centered, draw=black, fill=blue!30]
\tikzstyle{input} = [text centered]
\tikzstyle{arrow} = [thick,->,>=stealth]
\tikzstyle{output} = [text centered]
\tikzstyle{condition} = [text centered, font=\small\itshape]

\usepackage{threeparttable}

\theoremstyle{thmstyleone}%
\newtheorem{theorem}{Theorem}
\newtheorem{proposition}[theorem]{Proposition}%
\theoremstyle{thmstyletwo}%
\theoremstyle{thmstylethree}%
\newtheorem{definition}{Definition}

\newcommand{\argmin}{\operatornamewithlimits{\mathrm{argmin}}} 
\newcommand{\nucA}{\ensuremath{\text{\sc A}}}
\newcommand{\nucC}{\ensuremath{\text{\sc C}}}
\newcommand{\nucG}{\ensuremath{\text{\sc G}}}
\newcommand{\nucU}{\ensuremath{\text{\sc U}}}
\newcommand{\pairs}{\ensuremath{\mathit{pairs}}\xspace}
\newcommand{\unpaired}{\ensuremath{\mathit{unpaired}}\xspace}
\newcommand{\MFE}{\ensuremath{\text{\rm MFE}}\xspace}
\newcommand{\UMFE}{\ensuremath{\text{\rm uMFE}}\xspace}
\newcommand{\NED}{\ensuremath{\text{NED}}\xspace}

\newcommand{\new}{\ensuremath{\text{new}}\xspace}
\newcommand{\prob}{\ensuremath{\text{prob}}\xspace}
\newcommand{\rt}{\ensuremath{\text{root}}\xspace}
\newcommand{\step}{\ensuremath{\text{\rm step}}\xspace}

\newcommand{\LP}{\ensuremath{\mathit{loops}}\xspace}

\newcommand{\contract}{\ensuremath{\mathrm{Intersection}}}

\newcommand{\D}{\ensuremath{\mathrm{\Delta}}}
\newcommand{\DG}{\ensuremath{\mathrm{\Delta} G^\circ}}
\newcommand{\DDG}{\ensuremath{\mathrm{\Delta\Delta} G^\circ}}

\newcommand{\blangle}{\ensuremath{\boldsymbol{\langle}}}
\newcommand{\brangle}{\ensuremath{\boldsymbol{\rangle}}}
\newcommand{\proj}{\ensuremath{\vdash}}

\DeclareMathOperator{\defeq}{\stackrel{\Delta}{=}}

\renewcommand{\vec}[1]{\ensuremath{\boldsymbol{{#1}}}\xspace}

\newcommand{\vecx}{\ensuremath{\vec{x}}\xspace}
\newcommand{\vecy}{\ensuremath{\vec{y}}\xspace}
\newcommand{\vecz}{\ensuremath{\vec{z}}\xspace}

\newcommand{\ystar}{\ensuremath{\vec{y^\star}}\xspace}

\newcommand{\setxi}{\ensuremath{\hat{X}}\xspace}

\newcommand{\loops}{\ensuremath{\mathit{loops}}\xspace}

\newcommand{\um}{undesignable motif\xspace}
\newcommand{\m}{motif\xspace}
\newcommand{\M}{\ensuremath{\vec{m}}\xspace}
\newcommand{\vecm}{\ensuremath{\vec{m}}\xspace}
\newcommand{\ipairs}{\ensuremath{\mathit{ipairs}}\xspace}
\newcommand{\bpairs}{\ensuremath{\mathit{bpairs}}\xspace}
\newcommand{\mstar}{\ensuremath{\vec{m^\star}}\xspace}

\renewcommand{\emptyset}{\varnothing}
\newcommand{\card}{\ensuremath{\mathit{card}}\xspace}

\newcommand{\yrival}{\ensuremath{\mathcal{Y}_{\text{rival}}}\xspace}
\newcommand{\pnode}{\ensuremath{N_{\text{parent}}}\xspace}
\newcommand{\chnode}{\ensuremath{N_{\text{child}}}\xspace}
\newcommand{\newnode}{\ensuremath{N_{\text{new}}}\xspace}

\newcommand{\ours}{FastDesign\xspace}
\newcommand{\highlight}[1]{\emph{#1}}

\begin{document}

\journaltitle{Journal Title Here}
\DOI{DOI HERE}
\copyrightyear{2022}
\pubyear{2019}
\access{Advance Access Publication Date: Day Month Year}
\appnotes{Paper}

\firstpage{1}


\title{Fast and Versatile RNA Design via Motif-level Divide-and-Conquer and Structure-level Rival Search}

\author[1]{Tianshuo Zhou\ORCID{0009-0008-4804-0825}}
\author[3,4,5]{David H. Mathews\ORCID{0000-0002-2907-6557}}
\author[1,2,$\ast$]{Liang Huang\ORCID{0000-0001-6444-7045}}

\authormark{Zhou et al.}

\address[1]{\orgdiv{School of EECS}}
\address[2]{\orgdiv{Dept.~of Biochemistry \& Biophysics}, \orgname{Oregon State University}, \country{USA}}
\address[3]{\orgdiv{Dept.~of Biochemistry \& Biophysics}}
\address[4]{\orgdiv{Center for RNA Biology}}

\address[5]{\orgdiv{Dept.~of Biostatistics and Computational Biology}, \orgname{University of Rochester Medical Center}, \country{USA}}

\corresp[$\ast$]{Corresponding author. \href{email:email-id.com}{liang.huang.sh@gmail.com}}

\received{Date}{0}{Year}
\revised{Date}{0}{Year}
\accepted{Date}{0}{Year}



\abstract{
\textbf{Motivation:} RNA design aims to identify RNA sequences that fold into a target secondary structure. This task is challenging in terms of computational efficiency. Most existing methods focus on either minimum free energy (\MFE)-based or ensemble-based metrics, leaving a gap for a unified approach that performs well across both. We introduce a fast and versatile RNA design algorithm inspired by our previous work on the undesignability of RNA structures and motifs (i.e., sets of contiguous structural loops). Our approach decomposes a target structure into a tree of sub-targets where each leaf node corresponds to a motif and each internal node corresponds to a substructure. We first design partial sequences for each motif, then these partial sequences are selectively and recursively combined via the \emph{cube pruning} strategy borrowed from computational linguistics, enabling effective optimization of ensemble-based metrics. Finally, a novel whole-structure rival search further refines sequences to suppress misfolded alternatives and enhance \MFE-based performance.  \\
\textbf{Results:} Our method is highly efficient and also achieves state-of-the-art results on native RNAsolo structures and the Eterna100 benchmark, excelling in both ensemble- and  \MFE-based metrics. Additionally, it substantially improves the design of long-structure benchmark derived from 16S rRNA, increasing average folding probability from 0.18 to 0.39 with an order-of-magnitude speedup, demonstrating its effectiveness and scalability. \\
\textbf{Availability:}  Source code and data are available at: \url{https://github.com/shanry/FastDesign}. \\
\textbf{Supplementary information:} Supplementary text and data are available in a separate PDF. \\
\textbf{Contact:} \href{liang.huang.sh@gmail.com}{liang.huang.sh@gmail.com}
} 

\maketitle

\section{Introduction}

As structures dictate functions in structural biology, it is crucial to study the relations between biological sequences and their structures for non-coding RNAs~\cite{fiannaca+:2017nrc}.
Given a target structure, RNA design aims to find RNA sequences that can fold into that structure. In other words, RNA design is the inverse problem of RNA folding.
However, RNA design is considered NP-hard~\cite{bonnet+:2020designing}.

While numerous RNA design methods~\cite{portela:2018unexpectedly,rubio2018multiobjective,rubio+:2023eM2dRNAs,zhou+:2023samfeo} have been developed, most recent works have been focused on improving the \MFE-based metrics, i.e., how many of the structures (or puzzles) are solved by the \MFE or unique \MFE (\UMFE) criterion. 
Despite those efforts, existing RNA design methods have limitations in computational efficiency, design quality and technical explainability. 

First, optimization-based methods for RNA design rely on iteratively evaluating (folding) and mutating RNA sequences. The main time cost comes from folding entire RNA sequences, which has a cubic time complexity. As target structures get longer, the running time will increase rapidly.
Secondly, most RNA methods focus on the MFE criterion, i.e., finding sequences such that the target structure has minimum free energy (\MFE). 
However, this often causes a very low equilibrium probability for the target structure.
Last but not least, existing methods provide limited explainability on why some structures are hard or impossible to design. 

To address the above drawbacks, we introduce a fast and versatile RNA design method. This work is inspired by two insights from the works on RNA undesignability~\citep{zhou+:2026theory,zhou+:2025scalable,zhou+:2024undesignable,yao:2021thesis}: (1) The whole structure designability is restricted by the designability of local structural motifs. (2) \highlight{Rival structures} can help determine the necessary conditions of successful \MFE designs.
Our approach first leverages local designability to decompose a target structure into a tree of subtargets. We then adapt our earlier structure-level design algorithm, SAMFEO, to generate partial sequences for target motifs on leaf nodes. These partial sequences are recursively combined to form larger RNA sequences. To avoid the combinatorial explosion, we borrow the idea of \emph{cube pruning} from computational linguistics~\cite{huang+chiang:2007} to efficiently explore combinatorial search space. Such a motif-level divide-conquer-combine framework (Section~\ref{sec:motif-level}) enables effective optimization of ensemble-based metrics. Finally, a novel structure-level rival search (Section~\ref{sec:rival-search}) refines the resulting sequences to suppress misfolded alternatives and enhance \MFE-based performance.  
Our main contributions are:
\begin{enumerate}
\item \textbf{Motif-level divide-conquer-combine framework.} We propose a general motif-level divide-conquer-combine framework to enable effective ensemble optimization for RNA design. 
The framework is grounded in designability theory, leading to an interpretable and principled decomposition strategy.
To efficiently combine local designs into larger ones, we introduce a \emph{cube pruning} algorithm that efficiently explores the combinatorial design space with near-optimal efficiency.

\item \textbf{Structure-level rival search for \MFE-based design.} We develop a novel structure-level rival search algorithm to enhance \MFE-based evaluation metrics. 
The algorithm is driven by energy-difference calculations between the target structure and its rival structures, enabling aggressive pruning of the \MFE design space and substantially improving design effectiveness.

\item \textbf{Unified and scalable RNA design pipeline.}
We unify the motif-level divide-conquer-combine framework and the structure-level rival search algorithm into a single, end-to-end RNA design pipeline, as illustrated in Fig.~\ref{fig:framework}.
Our implementation is fast, flexible, and versatile: it achieves state-of-the-art ensemble-based performance, solves the largest number of RNAsolo structures and Eterna100 puzzles under \MFE-based criteria, and demonstrates clear advantages in both scalability and design quality on long structures ($>1000$~\emph{nt}).
\end{enumerate}

\vspace{-0.5cm}

\begin{figure*}[ht]
    \centering
    \vspace{-0.2cm}
    \includegraphics[width=\textwidth]{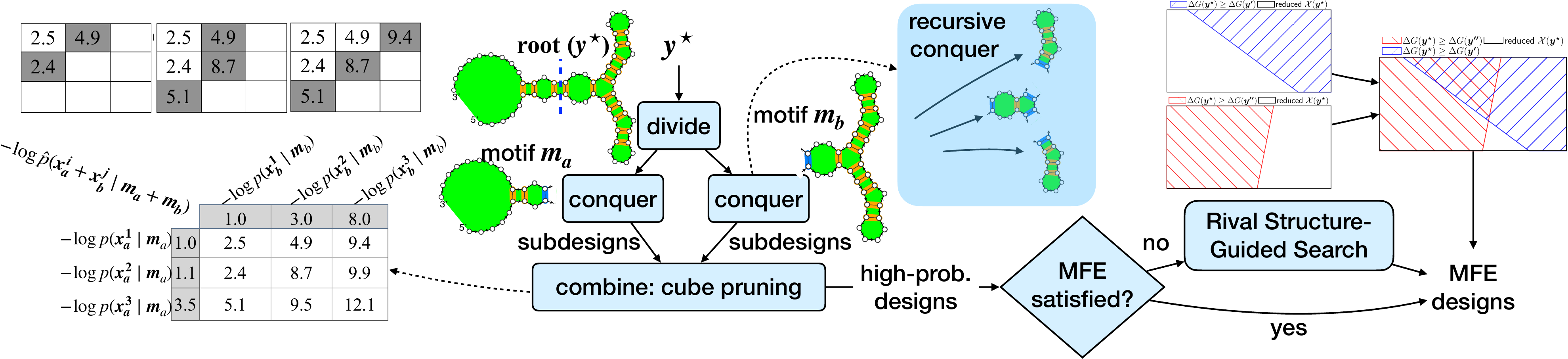}
    \caption{Overview of our unified RNA design framework. The method first applies a divide–conquer–combine strategy to generate high-probability designs, followed by a rival-structure–guided search when the MFE criterion is not satisfied, ultimately yielding MFE designs.} 
    \label{fig:framework}
\end{figure*}

\vspace{-0.25cm}

\section{Preliminaries on RNA Design}
\subsection{RNA Secondary Structure}
\begin{figure}[htbp]
    \centering
    \vspace{-0.2cm}
    \includegraphics[width=0.8\linewidth]{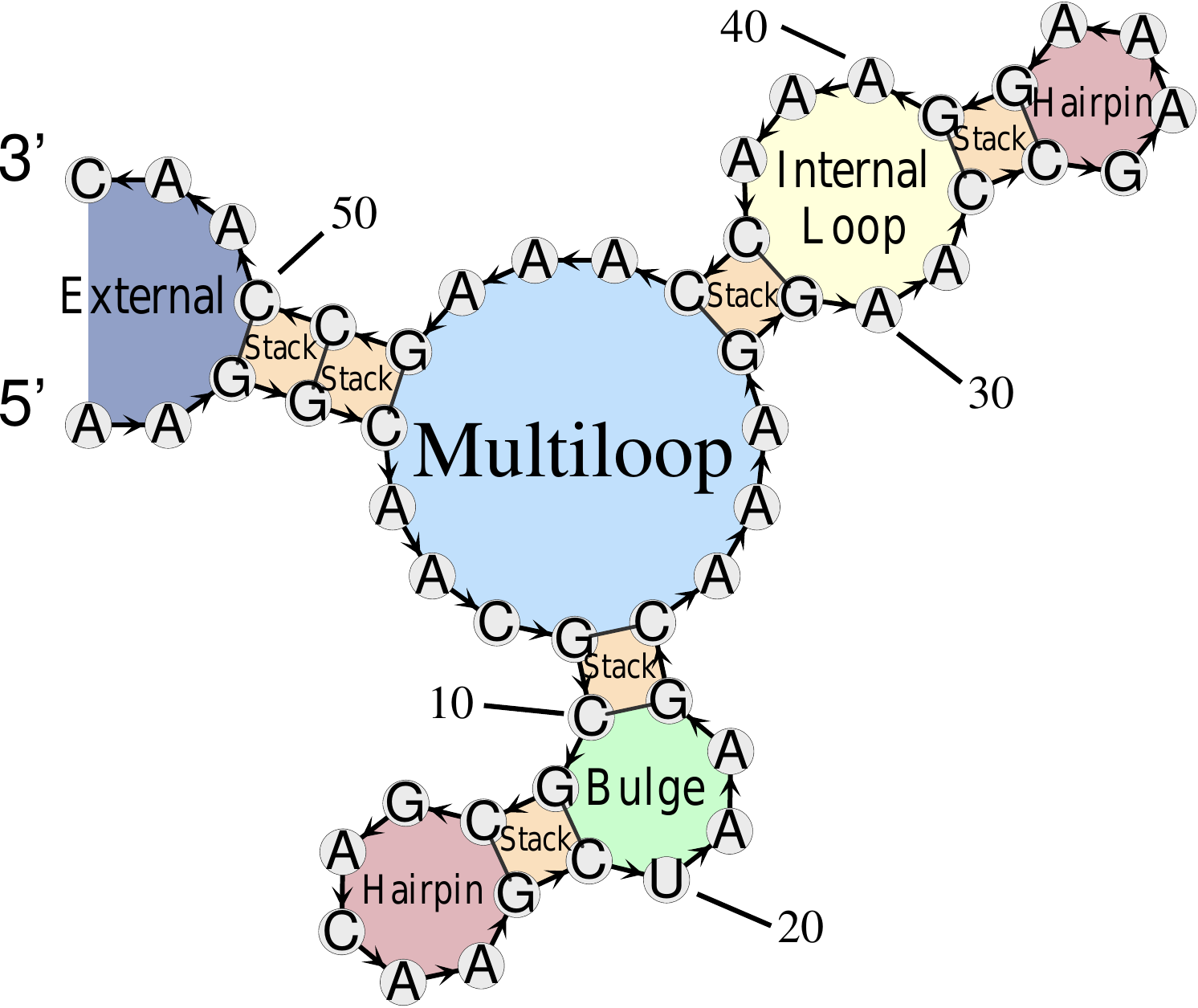}
    \caption{An example of secondary structure and loops.} 
    \label{fig:loops}
\end{figure}
An RNA sequence $\vecx$ of length $n$ is specified as a string of  base nucleotides $\vecx_1\vecx_2\dots \vecx_n,$  where $\vecx_i \in \{\nucA, \nucC, \nucG, \nucU\}$ for $i=1, 2,...,n$. The length of $\vecx$ can also be denoted as $|\vecx|$.
A secondary structure~$\mathcal{P}$  for $\vecx$ is a set of paired indices where each pair $(i, j) \in \mathcal{P}$ indicates two distinct bases $\vecx_i \vecx_j \in \{\nucC\nucG,\nucG\nucC,\nucA\nucU,\nucU\nucA, \nucG\nucU,\nucU\nucG\}$ and each index from $1$ to $n$ can only be paired once. A secondary structure is pseudoknot-free if there are no two pairs $(i, j)\in \mathcal{P}\text{ and }(k, l)\in~\mathcal{P}$ such that  $i<k<j<l$.  In short,  a pseudoknot-free secondary structure is a properly nested set of pairings in an RNA sequence.  Alternatively, $\mathcal{P}$ can be represented as a string~$\vecy=\vecy_1\vecy_2\dots \vecy_n$,  where a pair of indices $(i, j) \in~\mathcal{P}$ corresponds to $\vecy_i=``("$, $\vecy_j=``)"$ and any unpaired index $k$ corresponds to $\vecy_k=``."$. The unpaired indices in $\vecy$ are denoted as $\unpaired(\vecy)$ and the set of paired indices in $\vecy$ is denoted as $\pairs(\vecy)$, which is equal to~$\mathcal{P}$. 
Here we do not consider pseudoknots.

The \emph{ensemble} of an RNA sequence $\vecx$  is the set of all secondary structures that  $\vecx$ can possibly fold into, denoted as $\mathcal{Y}(\vecx)$. The \highlight{free energy (change)} $\DG(\vecx, \vecy)$ is used to characterize the stability of $\vecy \in \mathcal{Y}(\vecx)$. The lower the free energy change,~$\DG(\vecx, \vecy)$, the more stable the secondary structure $\vecy$ for $\vecx$. 

A structure \vecy is composed of a set of loops denoted as $\LP(\vecy)$, as shown in Fig.~\ref{fig:loops}. See Supplementary Section~\ref{sec:loops} for detailed explanation of different loop types. 
 A \highlight{motif} is defined as a contiguous set of loops in a structure (See Supplementary Section~\ref{sec:motif} and our prior work~\cite{zhou+:2024undesignable} for details).
 A \highlight{substructure} is a special motif that can be represented as a substirng of the dot-bracket string of a structure.
 The free energy change of a secondary structure $\vecy$ is the sum of the free energy change of each loop, i.e., $\DG(\vecx, \vecy) = \sum_{\vecz \in \loops(\vecy)} \DG(\vecx, \vecz)$, where each term $ \DG(\vecx, \vecz)$ is the energy for loop \vecz~\cite{mittal+:2024nndb}. 
The structure with the \highlight{minimum free energy} is  the most stable structure in the ensemble $\mathcal{Y}(\vecx)$. The minimum free energy of $\mathcal{Y}(\vecx)$ is denoted as $\MFE(\vecx)$ and defined as $\MFE(\vecx)=\min_{\vecy \in \mathcal{Y}(\vecx)} \DG(\vecx, \vecy)$.

\subsection{\MFE-based RNA Design}
 A structure $\ystar$ is an \MFE structure of $\vecx$ if and only if
\begin{equation}
 ~\forall \vecy \in \mathcal{Y}(\vecx)  \text{ and }\vecy \ne \ystar , \DG(\vecx, \ystar) \leq  \DG(\vecx, \vecy). \label{eq:mfe}
\end{equation}

Given a target structure $\ystar$, \MFE-based RNA design aims to find suitable RNA sequence $\vecx$ such that $\ystar$ is an \MFE structure of $\vecx$. For convenience, we define $\mathcal{X}(\vecy)$ as the set of all RNA sequences whose ensemble contains $\vecy$, i.e., $\mathcal{X}(\vecy)=\{ \vecx \mid \vecy \in \mathcal{Y}(\vecx) \}$. Here we have a more strict definition of \MFE criterion following previous studies~\cite{zhou+:2023samfeo}, i.e., $\vecx$ is a correct design if and only if $\vecy$ is the only \MFE structure of $\vecx$, which we call unique \MFE(\UMFE) criterion. Formally, $\vecy$ is the $\UMFE$ structure of \vecx if and only if
\begin{equation}
 ~\forall \vecy \in \mathcal{Y}(\vecx)  \text{ and }\vecy \ne \ystar , \DG(\vecx, \ystar) <  \DG(\vecx, \vecy). \label{eq:umfe}
\end{equation}

From the perspective of optimization, the satisfaction of \UMFE criterion requires that the structural distance between target structure $\ystar$ and \UMFE structure of $\vecx$ is minimized to $0$. 
The structural distance between two structures $\vecy'$ and $\vecy''$ is defined as
$d(\vecy', \vecy'') = n - 2\cdot |\pairs(\vecy')\cap \pairs(\vecy'')|  - |\unpaired(\vecy')\cap \unpaired(\vecy'')|, \label{eq:dist}$
where $\vecy'$ and $\vecy''$ have the same length $|\vecy'|=|\vecy''|=n$.

\subsection{Ensemble-based RNA Design}

\subsubsection{Equilibrium probability}
However, structure distance is not able to capture the equilibrium probability of a sequence folding into the target structure, which is defined based on the partition function, $Q(\vecx)$,
\begin{equation}
p(\vecy \mid \vecx) = \frac{e^{-\DG(\vecx, \vecy)/RT}}{Q(\vecx)} = \frac{e^{-\DG(\vecx, \vecy)/RT}}{\sum_{\vecy' \in \mathcal{Y}(\vecx)}e^{-\DG(\vecx, \vecy')/RT}}. \label{eq:prob}
\end{equation}
\subsubsection{Ensemble Defect}
The \emph{ensemble defect} is proposed to probabilistically sum up the structure distances between the target structure and other structures in the ensemble~\cite{Dirks+:2004, Zadeh+:2010}. Ensemble defect can be normalized to between 0 and 1 (\emph{normalized ensemble defect} or NED),
\begin{equation}
\hspace{-1.3cm}
\begin{aligned}
\NED(\vecx, \ystar) 	& = \frac{1}{n} \sum_{\vecy\in \mathcal{Y}(\vecx)} d(\ystar, \vecy) \cdot p(\vecy \mid \vecx). \label{eq:ned}
\end{aligned}\hspace{-.8cm}
\end{equation}
\vspace{-0.7cm}

\section{Motif-level Divide-Conquer-Combine}\label{sec:motif-level}

\begin{figure}[htbp]
	\centering
        \includegraphics[width=0.7\linewidth]{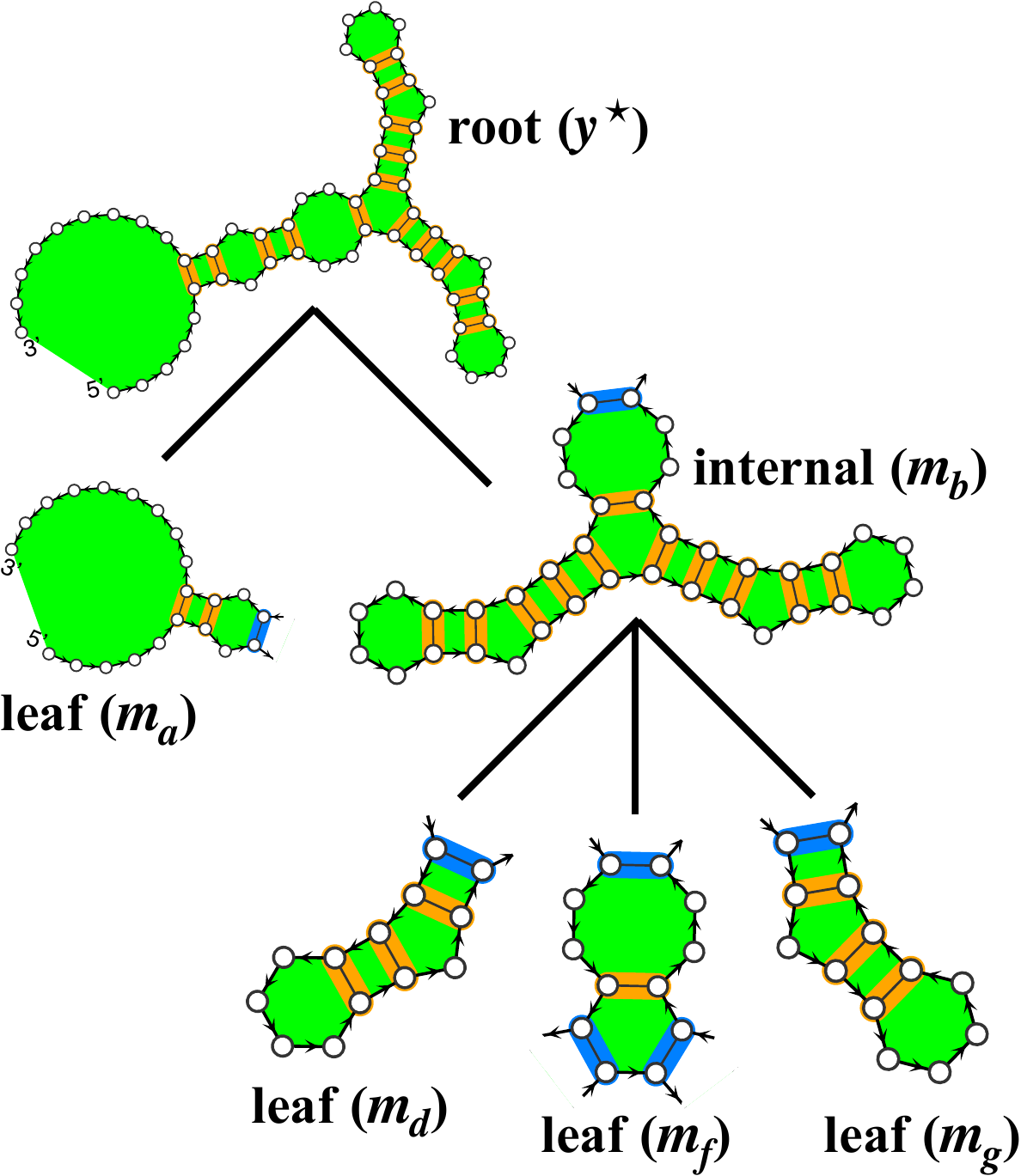}
        \caption{A target structure is decomposed into a tree of subtargets.}
         \label{fig:dccf}
\end{figure}
%

The divide-conquer-combine framework is illustrated in Fig.~\ref{fig:framework}.
Given a target RNA secondary structure, the framework recursively decomposes it into smaller subtargets (substructures or motifs), which are then designed and combined bottom-up.
Internal nodes correspond to larger substructures whose designs are obtained by combining the designs of their children,
while Leaf nodes correspond to motifs and are designed from scratch using an adapted version of SAMFEO~\cite{zhou+:2023samfeo} (Section~\ref{sec:conquer}), 

The complete procedure is summarized in Supplementary Algorithm~\ref{alg:dcc}, which invokes Supplementary Algorithms~\ref{alg:samfeo-motif} and~\ref{alg:samfeo}.

\subsection{Divide: Designability-driven  Decomposition}\label{sec:decompose}


In this subsection, we introduce a \emph{designability-driven decomposition} strategy, inspired by recent advances in the study of undesignable RNA motifs~\citep{yao:2021thesis,zhou+:2024undesignable,zhou+:2025scalable}.
The key idea is to decompose a structure at locations where the interaction between subparts is likely to be stable, thereby minimizing adverse coupling effects during recombination.

\begin{figure}[htbp]
    \begin{subfigure}[b]{0.20\textwidth}
        \includegraphics[width=\textwidth]{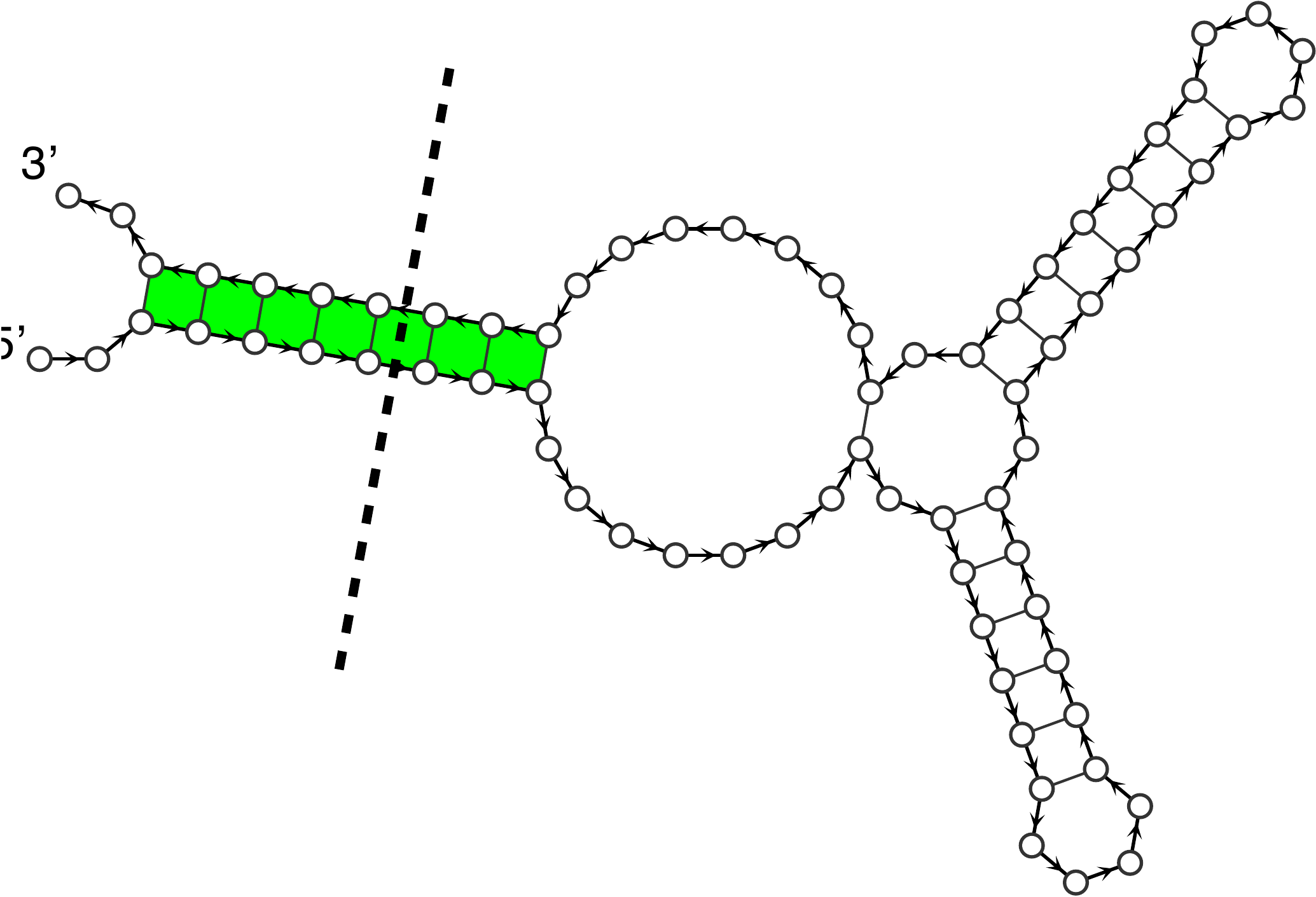}
        \vspace{-0.1cm}
        \caption{good split}
        \label{fig:eterna86}
    \end{subfigure}
    \begin{subfigure}[b]{0.13\textwidth}
        \includegraphics[width=\textwidth]{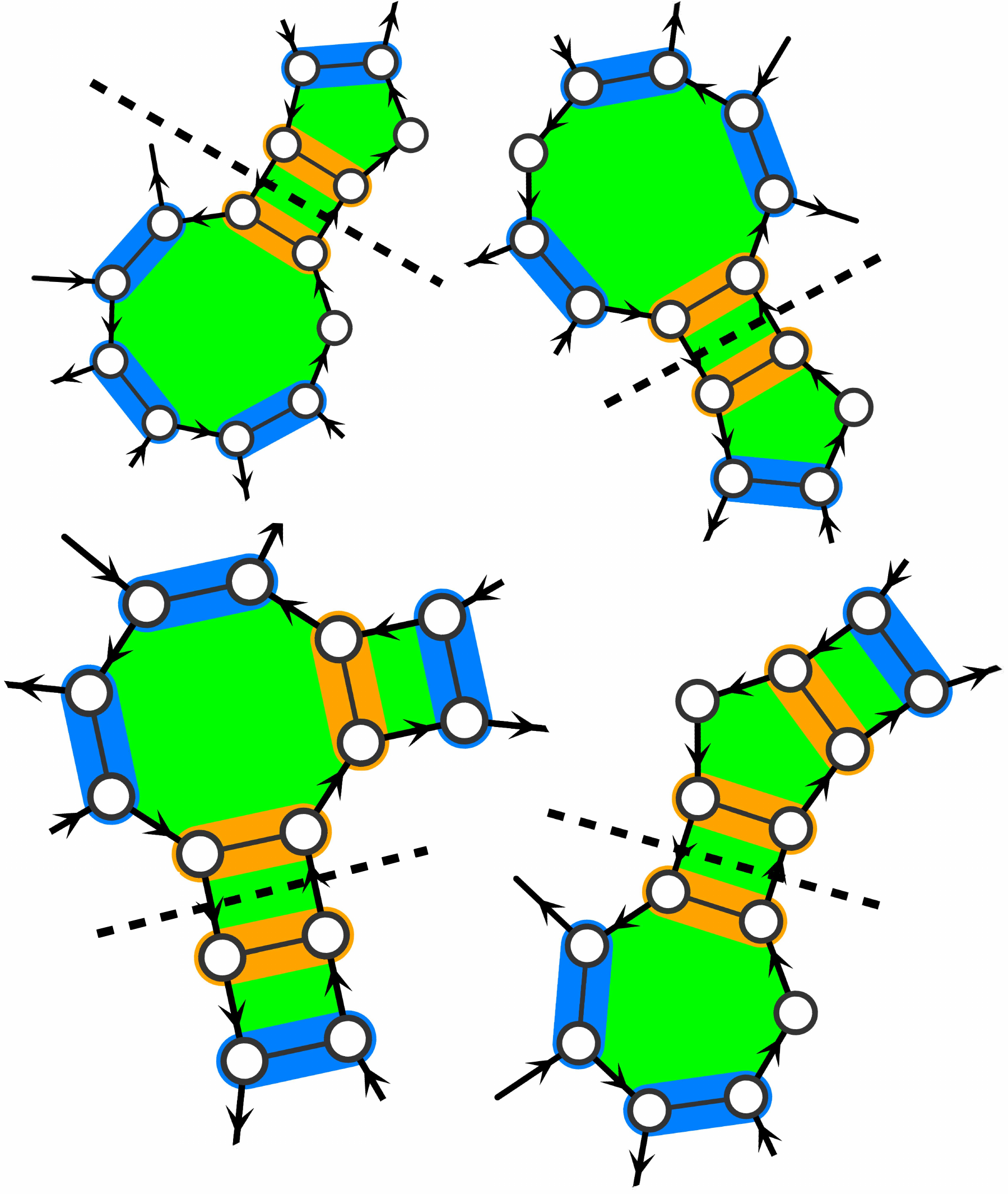}
        \vspace{-0.2cm}
        \caption{}
        \label{fig:motifs2x2}
  \end{subfigure}
  \begin{subfigure}[b]{0.10\textwidth}
  	\centering
        \includegraphics[width=0.8\textwidth]{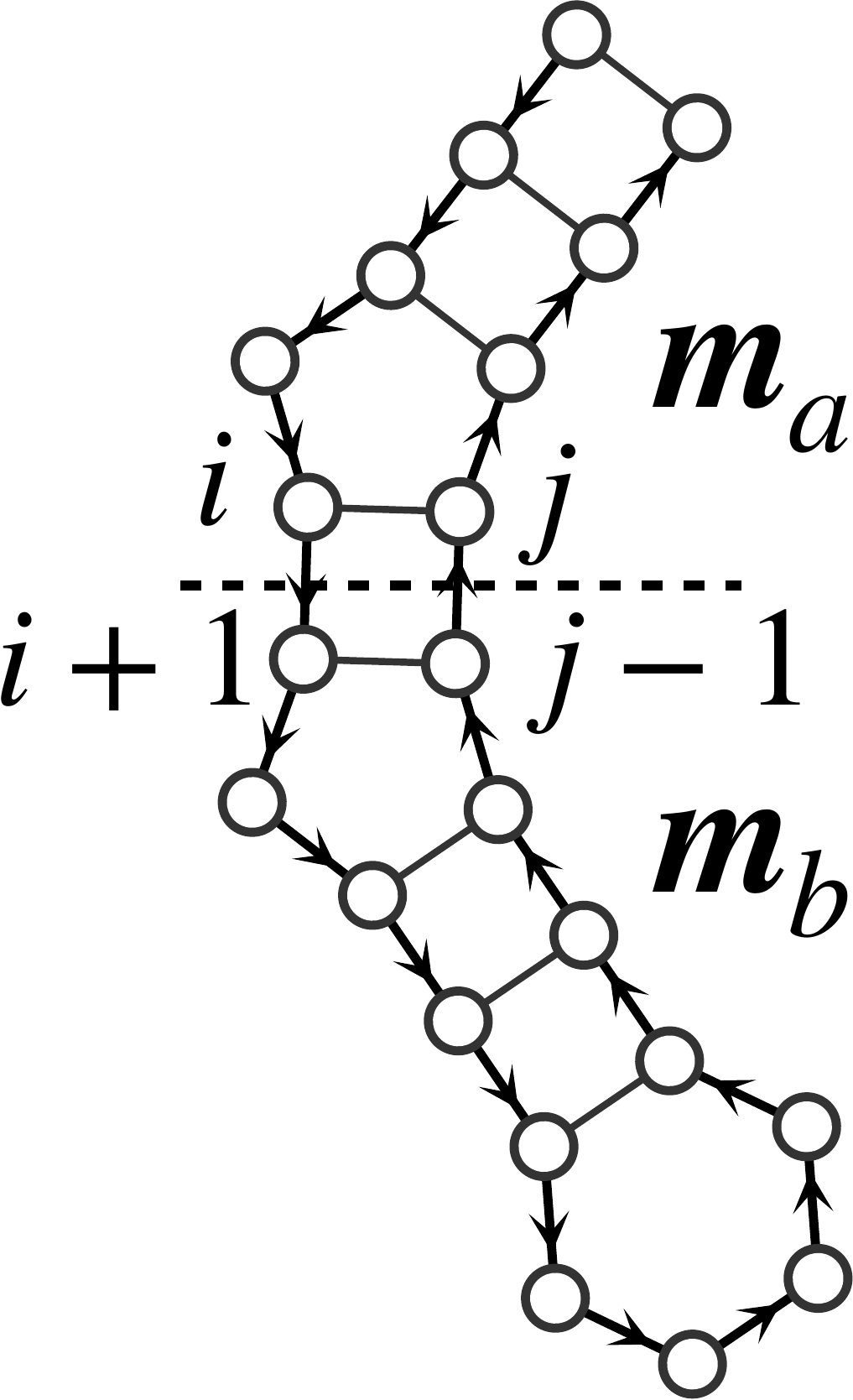}
        \vspace{-0.2cm}
        \caption{}
        \label{fig:motifab}
  \end{subfigure}
  \caption{(a) Bad split;  (b) Good split; (c) More \emph{easy-to-design} motifs; (d) A motif split into $\vecm_a$ and $\vecm_b$.}
\end{figure}


Without loss of generality, consider splitting a structure $\vecy$ (or a structural motif $\vecm$) into two submotifs $\vecm_a$ and $\vecm_b$, as shown in Fig.~\ref{fig:motifab}, denoted by
$\vecy = \vecm_a + \vecm_b$ (or $\vecm = \vecm_a + \vecm_b$).
Partial designs $\vecx_a$ and $\vecx_b$ can then be concatenated to form a complete design
$\vecx = \vecx_a + \vecx_b$.

The goal of such a decomposition is to achieve a high global folding probability
$p(\vecm_a + \vecm_b \mid \vecx_a + \vecx_b)$
by combining locally successful designs with high probabilities
$p(\vecm_a \mid \vecx_a)$ and $p(\vecm_b \mid \vecx_b)$.

Ideally, if the objective $-\log p(\vecm \mid \vecx)$ were perfectly additive, i.e.,
$$-\log p(\vecm_a +  \vecm_b \mid \vecx_a + \vecx_b)=-\log p(\vecm_a \mid \vecx_a) - \log p(\vecm_b \mid \vecx_b),$$
then high-quality local designs would always yield a high-quality global design.
However, this linearity does not hold in the Turner RNA folding model.
Instead, we observe a systematic one-sided deviation, formalized below.

\begin{theorem}\label{theo:delta}
\begin{equation}
\begin{aligned}
 &  -\log {p} (\vecm_a +  \vecm_b \mid \vecx_a + \vecx_b) \\
= &- \!\log {p} (  \!\vecm_a  \!\mid  \!\vecx_a)  \!- \!\log {p} ( \vecm_b  \!\mid  \!\vecx_b) + \delta \text{ for some } \delta > 0, \label{eq:delta}
 \end{aligned}
 \end{equation}
where $\delta$ quantifies the \emph{risk} induced by the decomposition $\vecm = \vecm_a + \vecm_b$.
\end{theorem}

\begin{proof}
Assume that $\vecm$ is decomposed at the stacking pairs
$\langle (i,j), (i+1,j-1) \rangle$, as illustrated in Fig.~\ref{fig:motifab}.
Let $P_{i,j}$ denote the event that bases $i$ and $j$ are paired.
By the chain rule,
\begin{equation}
\hspace{-1cm}
\begin{aligned}
&{p} ( \vecm_a +  \vecm_b \mid \vecx_a + \vecx_b )\\
 =&  {p} (P_{i,j}, P_{i+1,j-1} \mid \vecx_a + \vecx_b) \\
&\times {p} ( \vecm_a | \vecx_a, P_{i,j}, P_{i+1,j-1})  \times {p} ( \vecm_b | \vecx_b, P_{i,j}, P_{i+1,j-1}) \\
=& {p} (P_{i,j}, P_{i+1,j-1} \mid \vecx_a + \vecx_b) \times {p} ( \vecm_a \mid \vecx_a)  \times {p} ( \vecm_b \mid \vecx_b), \label{eq:product}
\end{aligned}\hspace{-1cm}
\end{equation}
Since boundary base pairs are enforced during motif-level folding and structures are pseudoknot-free,
the conditional probabilities reduce to $p(\vecm_a \mid \vecx_a)$ and $p(\vecm_b \mid \vecx_b)$.
Taking $-\log(\cdot)$ yields
\begin{equation}
\begin{aligned}
\hspace{-0.2cm}&-\log {p} ( \vecm_a +  \vecm_b \mid \vecx_a + \vecx_b ) \\
\hspace{-0.2cm}=& \underbrace{- \!\log p(P_{i,j}, \!P_{i+1,j-1} \!\mid \vec{x}^a \!+ \vec{x}^b)}_{\delta}  \!- \!\log {p} (  \!\vecm_a  \!\mid  \!\vecx_a)  \!- \!\log {p} (  \!\vecm_b  \!\mid  \!\vecx_b). \label{eq:proof}
\end{aligned}
\vspace{-0.7cm}
\end{equation}
\hfill
\end{proof}

\textbf{Interpretation.} The quantity $\delta$ captures the likelihood that the boundary stacking pairs connecting $\vecm_a$ and $\vecm_b$ fail to form.
A smaller $\delta$ indicates a more reliable decomposition.


Motivated by the physical interpretation of $\delta$, we prefer to decompose structures at stacking pairs that are likely to form with high probability.
To identify such locations, we construct a library of \emph{easy-to-design} motifs.

Specifically, we systematically enumerate all short motifs $\vecm_{\text{short}}$ of length at most 14~\emph{nt} and design\footnote{The procedure for designing motifs is described in subsection~\ref{subsc:basecase}.} them by maximizing
$p(\vecm_{\text{short}} \mid \vecx)$.
This probability serves as a surrogate measure of the stability of stacking interactions embedded in the motif.
Motifs that achieve $p(\vecm_{\text{short}} \mid \vecx) \ge 0.95$ are classified as \emph{easy-to-design}.
This process yields a library of over 6{,}000 motifs, with representative examples shown in Fig.~\ref{fig:motifs2x2}.

Given a target structure $\ystar$, we locate occurrences of easy-to-design motifs and decompose $\ystar$ at the corresponding stacking pairs.
Fig.~\ref{fig:eterna86} illustrates such a decomposition at a helical region, which naturally belongs to the easy-to-design motif class.

The full decomposition procedure is described in Algorithm~\ref{alg:decomp}.

\begin{algorithm}[ht]
\caption{Structure Decomposition}
\label{alg:decomp}
\begin{algorithmic}[1]
\State \Comment{Input \vecm: structure or motif}
\Function{Decompose}{$\vecm$, $N_{\text{parent}}=\text{None}$}
    \If{$\vecm_{\text{parent}}$ is None}\Comment{The root node}
    	\State $\pnode \gets \text{RootNode}(5', 3')$ \Comment{the whole structure}
    \EndIf
    \For{$\vecm^e \in M_{\text{easy}}$}\Comment{set of easy-to-design motifs}
    	\If{$\vecm^e \subseteq \vecm$}\Comment{$\vecm^e$ appears in \vecm}
    	   \State $\vec{m}_1,\! \vec{m}_2\!\gets\!$ {split} $\vecm$ at stack $\langle (i\!-\!1,\! j\!+\!1),\! (i,\! j)\rangle$ of  $\vecm^e$
	   \State $N_\text{new} \gets \text{NewNode}(i, j)$\Comment{New subtarget at$(i, j)$}
	   \For{$\chnode \in \pnode.children$}
	   	\If{$\chnode.\mathit{left} > i \text{ and } \chnode.\mathit{right} < j$}
			 \State $\chnode.parent \gets \newnode$
		\EndIf
	   \EndFor
	   \State \Call{Decompose}{$\vecm_2$, $N_\text{new}$}\Comment{Recursion}
	   \State \Call{Decompose}{$\vecm_1$, $\pnode$}\Comment{Recursion}
	\EndIf
	\State \Return \pnode
    \EndFor
    \State $N_\mathit{leaf} \gets \mathit{NewLeaf}(\vecm)$; $(N_\mathit{leaf}).\mathrm{parent} \gets \pnode$
    \State \Return \pnode
\EndFunction
\end{algorithmic}
\end{algorithm}

\subsection{Conquer: Base Cases and Recursion}\label{sec:conquer}
%
%
After decomposition, a target structure is represented as a tree of subtargets, where each leaf node corresponds to a motif and each internal node corresponds to a larger substructure, as illustrated in Fig.~\ref{fig:dccf}.
The design process proceeds bottom-up on this tree.

\subsubsection{Leaf Nodes (Motifs)}\label{subsc:basecase}
For each leaf node, corresponding to a motif $\mstar$, we design sequences from scratch.
To this end, we adapt our previous whole-structure RNA design algorithm, SAMFEO~\citep{zhou+:2023samfeo}, to operate at the motif level.
Given a target motif $\mstar$, the adapted SAMFEO algorithm searches for a sequence $\vecx$ that maximizes the local folding probability $p(\mstar \mid \vecx)$.
This provides a set of high-quality local designs that serve as base cases for the recursive combination process.
Details of the adaptation are provided in Supplementary Section~\ref{sec:samfeo}.

\subsubsection{Internal Nodes (Substructures)}
For each internal node, corresponding to a substructure composed of multiple child motifs or substructures, we construct candidate designs by combining the designs generated for its children.
Since the number of possible combinations grows combinatorially with the number of children, a naive exhaustive strategy is computationally infeasible.
Therefore, an efficient and principled combination strategy is required.

To address this challenge, we propose a \emph{cube pruning} approach, inspired by best-first decoding algorithms in machine translation~\cite{huang+chiang:2007}.
This approach enables efficient exploration of the combinatorial design space while prioritizing candidates that are likely to yield high global folding probabilities.
The method is described in detail in the following subsection.

\subsection{Combine: Cube Pruning}\label{sec:combine}
Without loss of generality, we first consider the case of combining two motifs, $\vecm_a$ and $\vecm_b$, to form their parent substructure
$\vecm_c = \vecm_a + \vecm_b$.
The generalization to more than two children follows naturally.
\begin{enumerate}
\item $\vecm_a$ has $k$ designs $X_a = [ {\vecx_a^1}, {\vecx_a^2}, \ldots, {\vecx_a^k} ]$, each with a known probability value $p(\vecm_a \mid \vecx_a^i)$, and the designs in $X_a$ are sorted by $p(\vecm_a \mid \vecx_a^i)$.
\item $\vecm_b$ has $k$ designs $X_b = \{ {\vecx_b^1}, {\vecx_b^2}, \ldots, {\vecx_b^k} \}$, each with a known probability value $p(\vecm_b \mid \vecx_b^j)$, and the designs in $X_b$ are sorted by $p(\vecm_b \mid \vecx_b^j)$.
\item $\vecm_c$ has $k^2$ candidates: the set of all combinations $X_a \times X_b = \{ \vecx_a^i + \vecx_b^j \mid \vecx_a^i \in X_a, \vecx_b^j \in X_b \}$.
\end{enumerate}
The objective is to select the top $k$ candidates from $X_a \times X_b$ with respect to the true folding probability.
A brute-force evaluation of all $k^2$ candidates (or $k^d$ for $d$ children) is infeasible for two reasons:
(i) the number of combinations grows exponentially with the number of submotifs;
(ii) evaluating each candidate requires RNA folding, which is computationally expensive.
Thus, the core challenge is to \emph{identify near-optimal combinations while minimizing the number of folding evaluations}.


\begin{figure}[htbp]
        \begin{subfigure}[b]{0.34\textwidth}
        \includegraphics[width=\textwidth, angle=0]{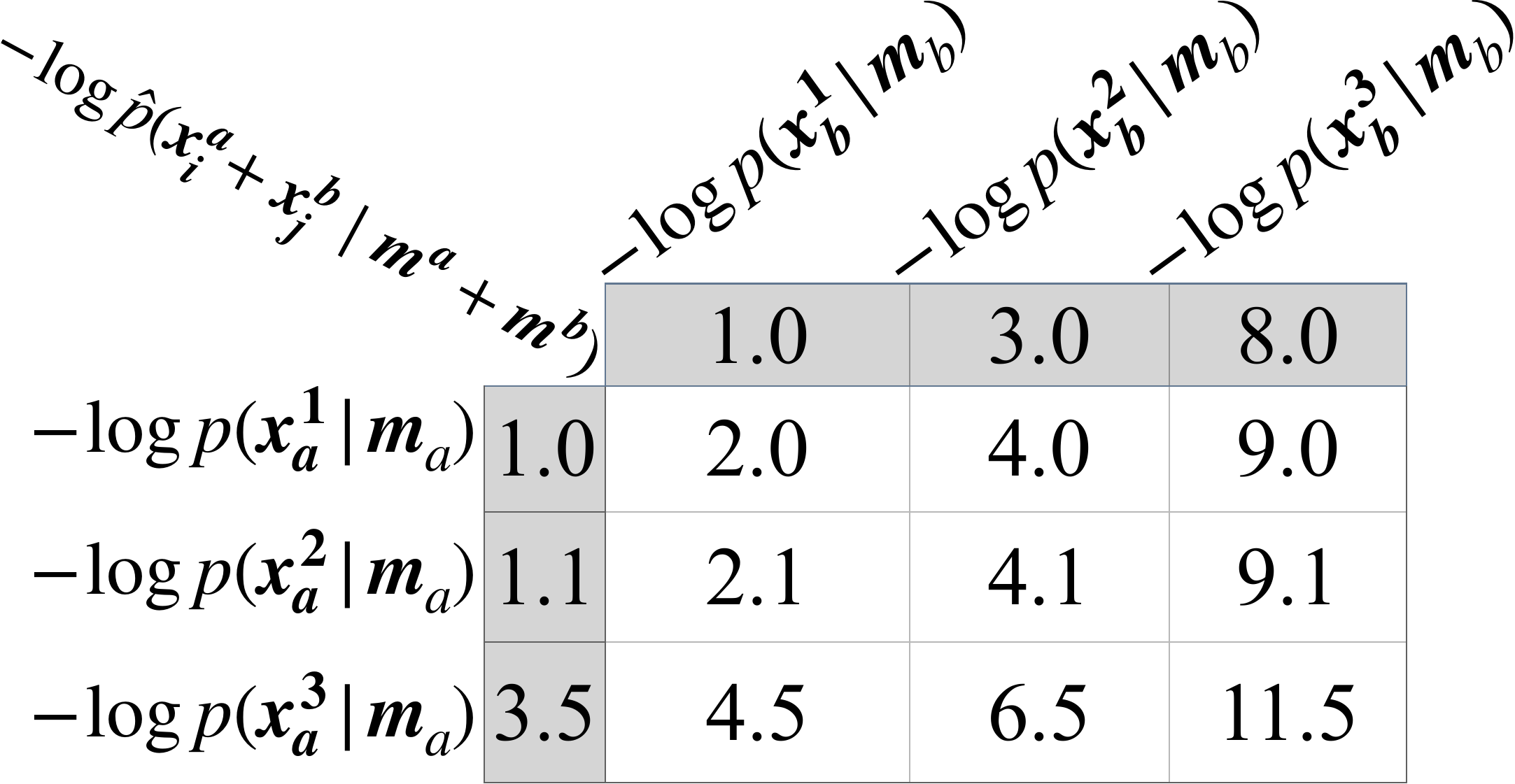}
        \caption{}
        \label{fig:cubic_approx}
    \end{subfigure}
    \hfill
    \begin{subfigure}[b]{0.15\textwidth}
        \includegraphics[width=\textwidth]{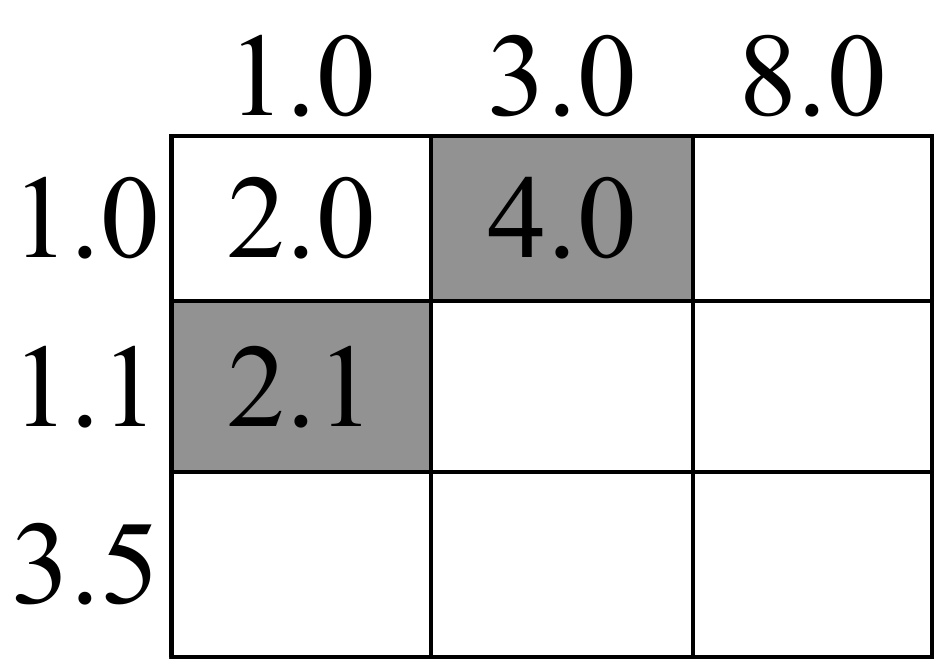}
        \caption{}
        \label{fig:prune_a1}
    \end{subfigure}
    \hfill
    \begin{subfigure}[b]{0.15\textwidth}
        \includegraphics[width=\textwidth]{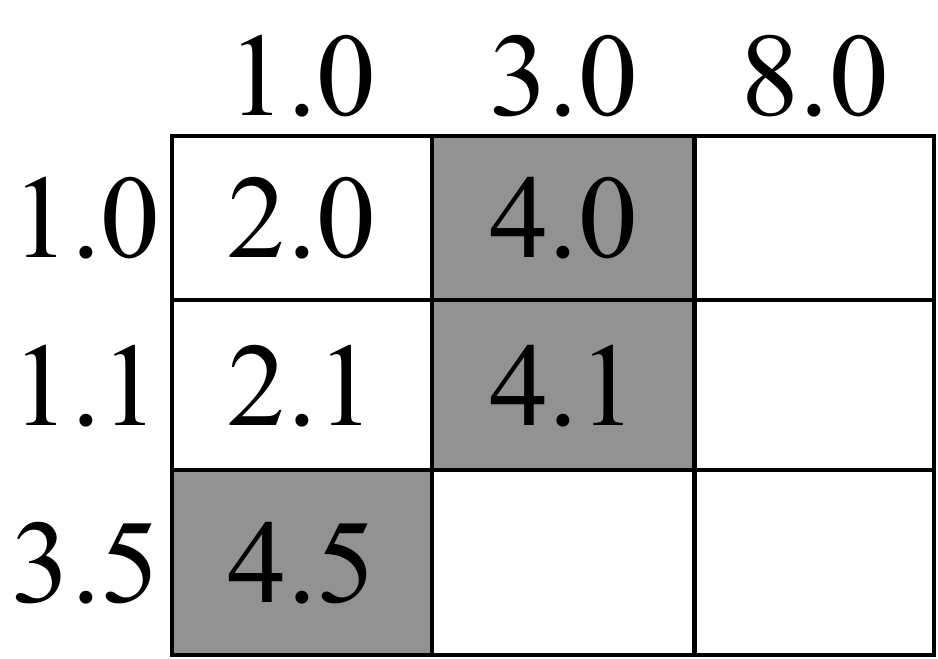}
        \caption{}
        \label{fig:prune_a2}
  \end{subfigure}
   \hfill
    \begin{subfigure}[b]{0.15\textwidth}
        \includegraphics[width=\textwidth]{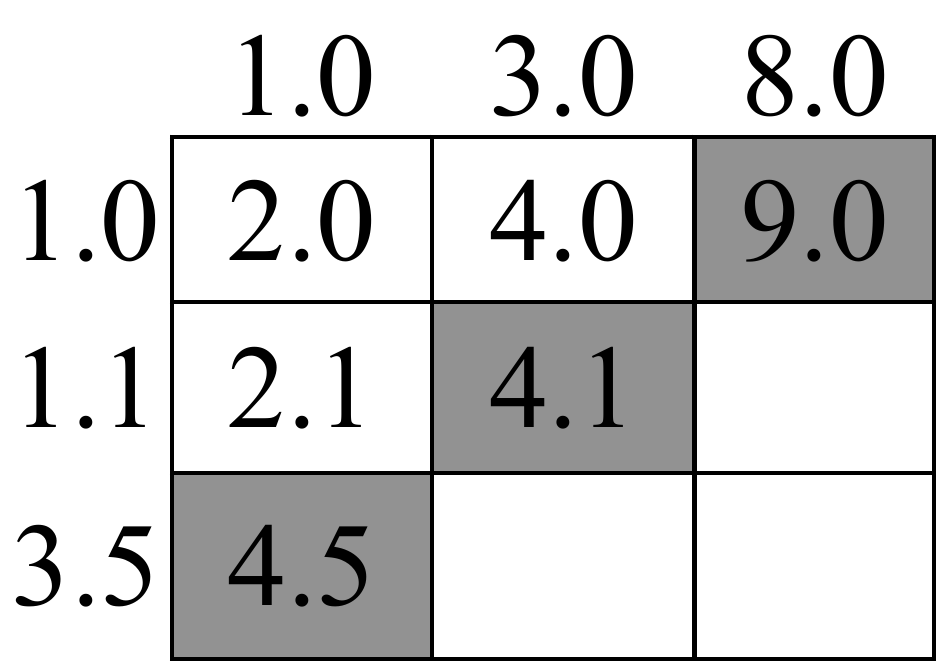}
        \caption{}
        \label{fig:prune_a3}
  \end{subfigure}
  \caption{Cube pruning over \textbf{approximated values}. (a): the numbers in the grid denote the approximated scores by linear addition; (b)-(d): the best-first enumeration of the top three items. Notice
that the items popped in (b) and (c) are out of order due to the non-monotonicity of the combination cost.}
\vspace{-0.1cm}
\end{figure}
\begin{figure}[htbp]
        \begin{subfigure}[b]{0.34\textwidth}
        \includegraphics[width=\textwidth, angle=0]{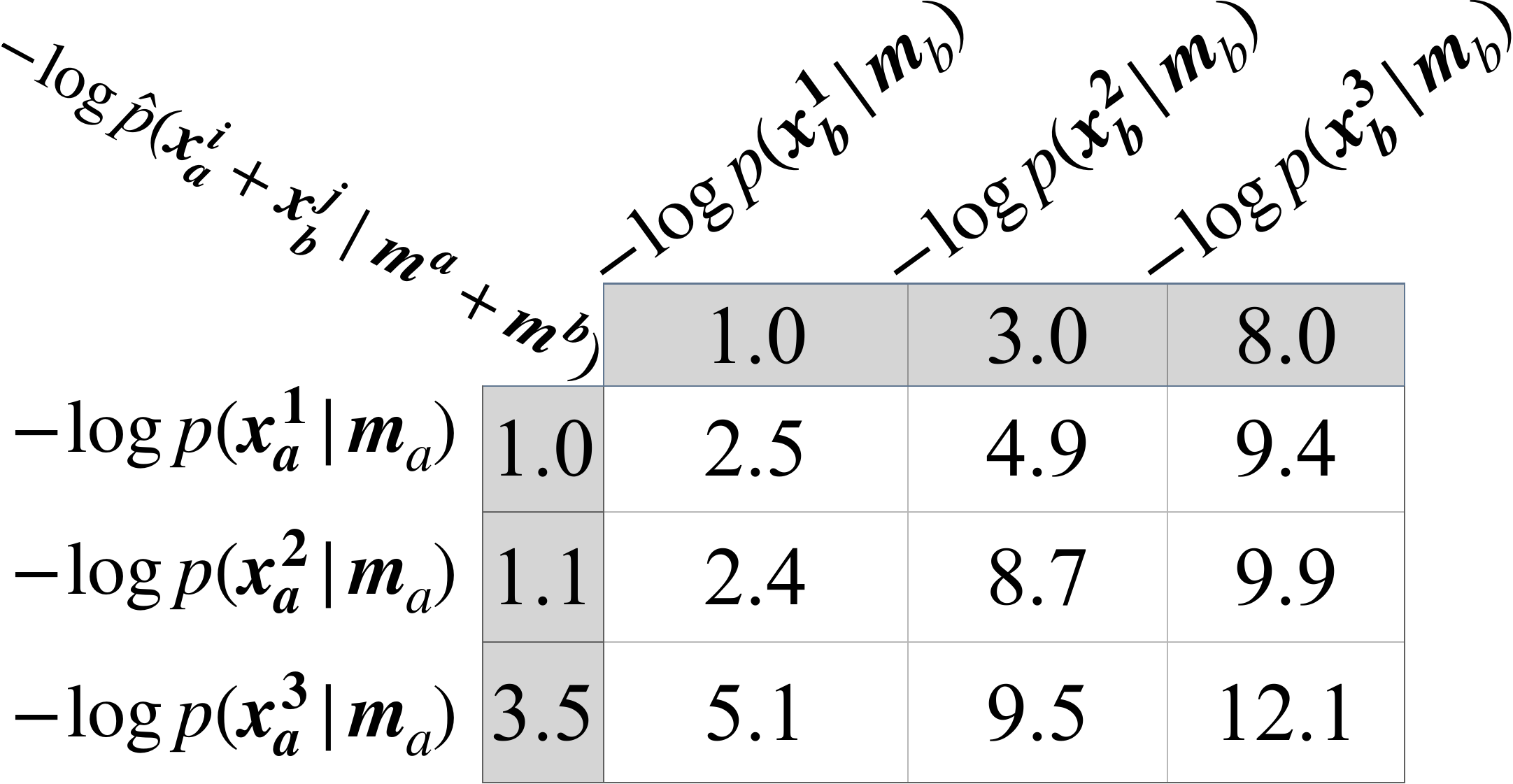}
        \caption{}
        \label{fig:cubic_real}
    \end{subfigure}
    \hfill
    \begin{subfigure}[b]{0.15\textwidth}
        \includegraphics[width=\textwidth]{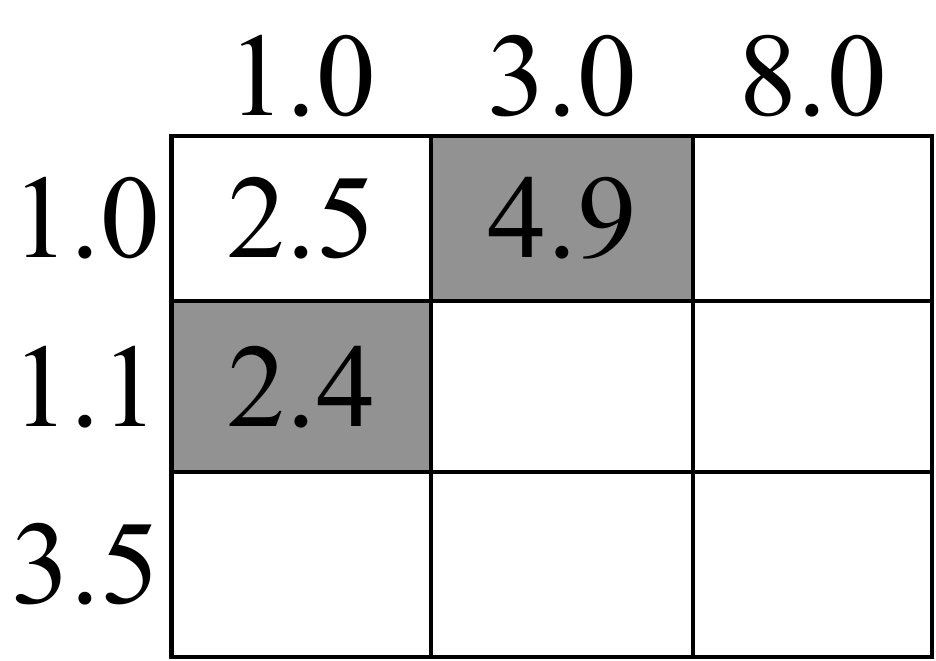}
        \caption{}
        \label{fig:prune_r1}
    \end{subfigure}
    \hfill
    \begin{subfigure}[b]{0.15\textwidth}
        \includegraphics[width=\textwidth]{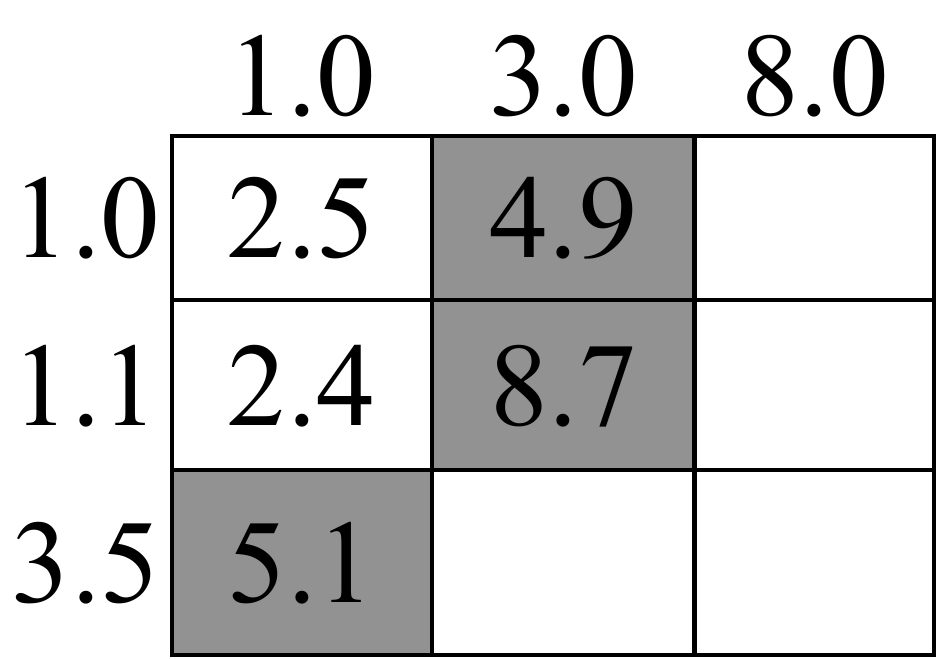}
        \caption{}
        \label{fig:prune_r2}
  \end{subfigure}
   \hfill
    \begin{subfigure}[b]{0.15\textwidth}
        \includegraphics[width=\textwidth]{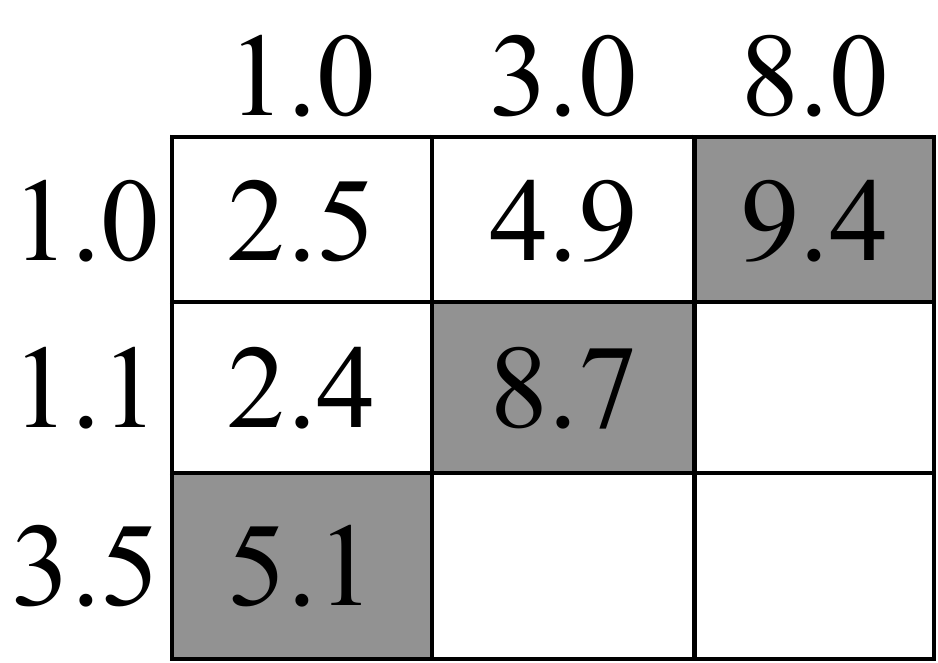}
        \caption{}
        \label{fig:prune_r3}
  \end{subfigure}
  \caption{ Cube pruning over \textbf{real values}. (a): the numbers in the grid denote the real scores by evaluation at the cost of RNA folding; (b)-(d): the best-first enumeration of the top three items. Notice
that the items popped in (b) and (c) are out of order due to the non-monotonicity of the combination cost.}
\end{figure}

If the objective were perfectly additive,  i.e., $-\log p(\vecm_a +  \vecm_b \mid \vecx_a^i + \vecx_b^j)=-\log p(\vecm_a \mid \vecx_a^i) -\log p(\vecm_b \mid \vecx_b^j)$, and the $k^2$ candidates would form a monotonically ordered two-dimensional grid
(Fig.~\ref{fig:cubic_approx}). In this idealized setting, the best candidate appears at the top-left corner, and the next-best candidates must be adjacent to it.

This property allows an efficient best-first enumeration using a priority queue: start from the top-left cell, repeatedly pop the current best item, and push its unvisited neighbors into the queue.
As shown in Figs.~\ref{fig:prune_a1}–\ref{fig:prune_a3}, this process maintains a frontier of the most promising candidates, ensuring the $k$ best items can be obtained after $k$ pops.

However, $\log p(\vecm \mid \vecx)$ does not satisfy the linearity as stated in Theorem~\ref{theo:delta}.
Computing the exact value of $\log p(\vecm_a + \vecm_b \mid \vecx_a^i + \vecx_b^j)$ requires RNA folding, which is computationally expensive.
To mitigate this cost, we approximate it by linear addition:
\vspace{-0.2cm}
\begin{equation}
\begin{aligned}
 &-\log \hat{p}(\vecm_a +  \vecm_b \mid \vecx_a^i + \vecx_b^j)\\
 \defeq &-\log p(\vecm_a \mid \vecx_a^i) - \log p(\vecm_b \mid \vecx_b^j). \label{eq:linear}
\end{aligned}
\end{equation}
This approximation provides an efficient heuristic that guides the search toward promising regions of the combinatorial space.

When we switch from approximated to real values (Fig.~\ref{fig:cubic_real}), the monotonic property of the grid no longer holds, as shown in Figs.~\ref{fig:prune_r1}–\ref{fig:prune_r3}.
For example, in Fig.~\ref{fig:prune_r1}, an item with score 2.5 is selected before one with a better (lower) score of 2.4, reflecting non-monotonicity of the true combination scores.
While this introduces occasional ranking errors, the overall accuracy remains satisfactory compared to the full combinatorial search, with substantially improved efficiency.

To compensate for this imperfection, we evaluate the \emph{true} folding probability of each candidate when it is popped from the priority queue.
This strategy incurs at most $2k$ folding evaluations in the two-motif case, which is negligible compared to the $k^2$ evaluations required by exhaustive search.
As a result, cube pruning achieves a substantial speedup while maintaining competitive solution quality.

Algorithm~\ref{alg:cubic_pruning} presents the cube pruning procedure for combining two motifs.
A generalized version for combining multiple motifs is provided in Algorithm~\ref{alg:cubic_pruning_multi}.

\begin{algorithm}[htbp]
\caption{\textsc{CubePruning}$(X_a, X_b, \vecm_a, \vecm_b, k)$ \Comment{simplified version of combining two motifs}}
\label{alg:cubic_pruning}
\begin{algorithmic}[1]
\Function{CubePruning}{$X_a, X_b, \vecm_a, \vecm_b, k$} 
    \State $\vecm_c \gets \vecm_a + \vecm_b$\Comment{combine motifs}
    \State Initialize a priority queue $\mathcal{Q}$
    \State Initialize empty set $X_c \gets \varnothing$
    \State Eval. true score $s = -\log p(\vecm_c \mid \vecx_a^1 + \vecx_b^1)$ by folding
    \State Push top-left combo $(1,1)$ into $\mathcal{Q}$ with score $s$
    
    \While{$|X_c| < k$ and $\mathcal{Q} \neq \varnothing$}
        \State $(i,j,s) \gets$ \Call{PopBest}{$\mathcal{Q}$} \Comment{lowest true score combo}
        \State Append $\vecx_a^i + \vecx_b^j$ to $X_c$ with score $s$
        \For{$(i', j') \!\in\! \{(i\!+\!1,j), (i,j\!+\!1)\}$ within bounds}
            \If{$(i', j')$ not yet visited}
                \State Eval. true score $s' = -\log p(\vecm_c \mid \vecx_a^{i'} + \vecx_b^{j'})$ 
                \State Insert $(i', j')$ into $\mathcal{Q}$ with score $s'$
            \EndIf
        \EndFor
    \EndWhile
    \State \Return $X_c$
\EndFunction
\end{algorithmic}
\end{algorithm}

\section{Structure-level Rival Search}\label{sec:rival-search}

The optimization objective of the above motif-level divide-conquer-combine stage is to maximize the folding probability $p(\ystar | \vecx)$. 
Although an RNA sequence $\vecx$ with a high $p(\ystar \mid \vecx)$ often satisfies the \MFE or \UMFE criterion, there are cases where a high probability does not guarantee either.
Consider two sequences $\vecx_1$ and $\vecx_2$ where $p(\ystar \mid \vecx_1) > p(\ystar \mid \vecx_2)$. 
It is still possible that $\vecx_2$ is the \UMFE solution for the target structure $\ystar$, whereas $\vecx_1$ is not. 
This occurs when another structure $\vecy'$ in the ensemble $\mathcal{Y}(\vecx_1)$ has slightly lower free energy (or slightly higher equilibrium probability) than $\ystar$.
In such cases, the \MFE structure of $\vecx_1$ is usually close to $\ystar$, and only minor modifications to $\vecx_1$ are needed to convert it into an \MFE or \UMFE solution. 
Based on this observation, we propose an efficient rival structure search algorithm to bridge the gap between ensemble-based and \MFE-based designs.


Consider the example in Fig.~\ref{fig:ex1}, where $\vecx$ is obtained by optimizing $p(\ystar \mid \vecx)$, and $\vecy'$ is its unique \MFE structure. 
Following our previous work~\citep{zhou+:2024undesignable}, we call $\vecy'$ a \highlight{rival structure} with respect to $\ystar$, since $\DG(\vecx, \ystar) > \DG(\vecx, \vecy')$. 
To modify $\vecx$ into $\vecx'$ that satisfies the \MFE or \UMFE criterion, a necessary (but not sufficient) condition is $\DDG(\vecx, \ystar, \vecy') = \DG(\vecx', \ystar) - \DG(\vecx', \vecy') < 0$.

%

According to prior analysis of undesignability~\citep{zhou+:2024undesignable}, the free energy difference ~$\DDG(\vecx, \ystar, \vecy')$
depends only on a few so-called \highlight{differential positions} in $\vecx$ (nine in this case), as annotated with `$*$' in Fig.~\ref{fig:ex1}. 
Since computing $\DDG$ is equivalent to a simple energy table lookup, we can efficiently (about 1 sec.) enumerate all possible nucleotide combinations on these positions (36,864 in total) to identify those satisfying $\DDG(\vecx, \ystar, \vecy') < 0$. 
Among them, only nine modifications (listed in Table~\ref{tab:dc1}) lead to $\DDG(\vecx, \ystar, \vecy') < 0$, and by testing these nine variants, we successfully obtained a \UMFE solution $\vecx'$ for $\ystar$.

\begin{figure}
\hspace{-0.3cm}
\setlength{\tabcolsep}{1pt}
  \begin{tabular}{rl}
\scriptsize $\ystar$& \tiny \verb|((.((..((..(...(.((.((....))))..)...(.((..((....))..))..)...).))..))))|\\
\scriptsize $\vecy'$&	\tiny \verb|((.((..((..(...(.((.((....))))..).....((..((....))..))......).))..))))|\\
\scriptsize $\vecx$\, & \tiny		\verb|GCACGGACUGACAAAGAGGACCGAAAGGCUGACAAACGGCGACCAAAAGGGAGCAAGGAAGAGGGACGGC| \\[-0.2cm]
          & 	\tiny	\verb|                                   ****              *****  |  \\[-0.2cm]
\scriptsize $\vecx'$ &	\tiny \verb|GCACGGACUGACAAAGAGGACCGAAAGGCUGACAAAGAGCGACCAAAAGGGAGUGACAAAGAGGGACGGC|
 \end{tabular}
  \caption{The \MFE structure of \vecx is $\vecy'$. \UMFE solution $\vecx'$ is obtained by modifying $\vecx$ on differential positions.}
  \label{fig:ex1}
\end{figure}

\begin{figure*}
  \begin{minipage}{.37\textwidth}
    \centering
    \captionof{table}{Design constraint induced  by $\vecy'$ in Fig.~\ref{fig:ex1}. $I$: indices or positions; $\hat{\vecx}^1-\hat{\vecx}^9$: nucleotides on $I$.}\label{tab:dc1}
    \begin{tabular}{c|ccccccccccc}
	$I$ & 36 & 37 & 38  & 39  & 54 & 55 & 56 & 57 & 58\\ \hline
	$\hat{\vecx}^1$&A & G & A & G & U & G & A & C & A\\ \hline
	$\hat{\vecx}^2$&C & G & A & G & U & G & A & C & A \\ \hline
	$\hat{\vecx}^3$&G & G & A & G & U & G & A & C & A \\ \hline
	$\hat{\vecx}^4$&U & G & A & G & U & G & A & C & A \\ \hline
	$\hat{\vecx}^5$&A & G & A & G & U & G & A & C & G \\ \hline
	$\hat{\vecx}^6$&C & G & A & G & U & G & A & C & G \\ \hline
	$\hat{\vecx}^7$&G & G & A & G & U & G & A & C & G \\ \hline
	$\hat{\vecx}^8$&U & G & A & G & U & G & A & C & G \\ \hline
	$\hat{\vecx}^9$&C & G & A & G & U & G & A & C & U  \\ \hline
	\end{tabular}
  \end{minipage}%
\hfill
  \begin{minipage}{.24\textwidth}
    \centering
    \vspace{-0.25cm}
    \includegraphics[width=\linewidth]{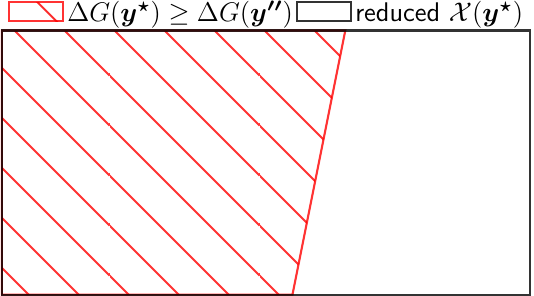}
    \caption{\mbox{$\vecy''$ reduces design space.}}\label{fig:ds1}
    \includegraphics[width=\linewidth]{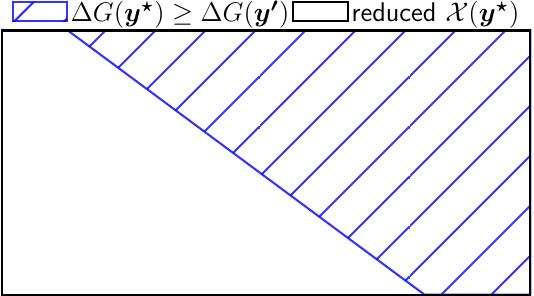}
    \caption{\mbox{$\vecy'$ reduces design space.}}\label{fig:ds2}
  \end{minipage}
  \hspace{0.1cm}
  \begin{minipage}{.25\textwidth}
    \centering
    \includegraphics[width=\linewidth]{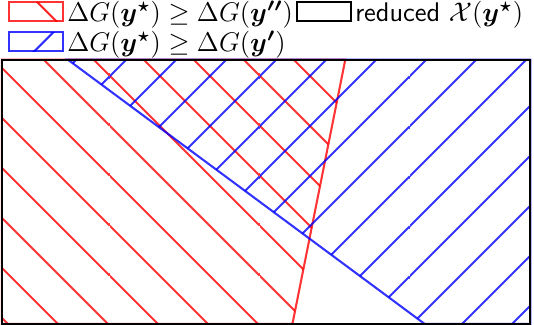}
    \caption{{Intersection of two reduced design spaces in \cref{fig:ds1,fig:ds2}.}}\label{fig:dsinter}
  \end{minipage}
\end{figure*}
The above search process, guided by a single rival structure $\vecy'$, does not always guarantee finding a \UMFE\ solution. 
We therefore generalize this idea by allowing multiple rival structures to jointly guide the search.
The key observation is that each rival structure partitions the RNA design space into two regions:

\begin{proposition}\label{prop:split}
 Given a target structure $\ystar$ and another structure $\vecy' \neq \ystar$, the RNA design space $\mathcal{X}(\ystar)$ can be divided into the two sets below.
 \begin{enumerate}
 \item \resizebox{\linewidth}{!} {$\mathcal{X}(\ystar<\vecy') = \{ \vecx \mid  \DG(\vecx, \ystar) <  \DG(\vecx, \vecy'), \vecx \in \mathcal{X}(\ystar) \} $};
 \item \resizebox{\linewidth}{!} {$\mathcal{X}(\ystar\geq\vecy') = \{ \vecx \mid  \DG(\vecx, \vecy') \leq \DG(\vecx, \ystar), \vecx \in \mathcal{X}(\ystar) \} $.}
 \end{enumerate}
 \end{proposition}
 
Since $\DDG (\vecx, \ystar, \vecy')$ is a necessary condition for \UMFE\ solutions, each rival structure $\vecy'$ effectively reduces the feasible design space $\mathcal{X}(\ystar)$, as illustrated in Figs.~\ref{fig:ds1} and~\ref{fig:ds2}. 
The reduced space can be represented as a design constraint, as shown in Table~\ref{tab:dc1}. 
When multiple rival structures (e.g., $\vecy'$ and $\vecy''$) are considered, their reduced spaces can be intersected (Fig.~\ref{fig:dsinter}) to further narrow the search region. 
Design space intersection can be achieved by combining their corresponding design constraints, as detailed in Supplementary Section~\ref{sec:intersection}.
In practice, we progressively identify multiple rival structures so that the design space is continuously reduced, efficiently guiding the search toward \MFE/\UMFE\ solutions, as detailed in Supplementary Algorithm~\ref{alg:rivalsearch}.

\vspace{-0.5cm}
\section{Experiments}

\subsection{Evaluation setting}
\subsubsection{Datasets}
Our experiments are conducted on a diverse set of RNA design benchmarks, including native structures, human-crafted structures, and structures derived from native RNA sequences.
The basic information on these datasets are listed below and structure selection details are in Supplementary Section~\ref{sec:datasets}.
\begin{enumerate}
\item \textbf{RNAsolo764.} RNAsolo~\citep{badura+:2025} is a comprehensive collection of RNA secondary structures curated from experimentally determined 3D structures~\citep{adamczyk+:2022rnasolo}. After preprocessing, 764 native RNA structures remained.
\item \textbf{ArchiveII100.} We select 100 native structures from the ArchiveII benchmark.~\citep{sloma+mathews:2016,Cannone+:2002}.
\item \textbf{Eterna100.} The Eterna100 benchmark~\citep{anderson+:2016} consists of 100 secondary structures designed by human players of the online game Eterna. 
\item \textbf{16S-MFE.} The 16S-MFE ~\citep{zhou+:2023samfeo} was introduced to assess the ability of RNA design methods to handle long structures.
It contains 10 \MFE structures of 16S rRNA in the ArchiveII benchmark.
\end{enumerate}
\vspace{0.1cm}
\textbf{Metrics} 
Following the settings in the literature~\citep{rubio2018multiobjective, anderson+:2016, zhou+:2023samfeo}, we use RNAfold from ViennaRNA 2.7~\cite{lorenz+:2011} as a folding engine, and each method is run 5 times to design a structure. We report the union and best metrics accordingly. 
Except for extra explanation, the following metrics are used for evaluation:
\begin{enumerate}
\item Number of solved puzzles by the \MFE or \UMFE criterion over 5 runs.
\item Best equilibrium probability~$p(\ystar \mid \vecx)$ or ensemble defect~$\NED(\vecx, \ystar)$   over 5 runs.
\item Running time per structure averaged over 5 runs.
\end{enumerate}
\vspace{0.1cm}
\textbf{Baselines}
We use following methods as baselines. 
\begin{enumerate}
\item RNAinverse,  default RNA design algorithm optimizing structure distance in ViennaRNA~\cite{lorenz+:2011}.
\item RNAinverse-pf, which is a variant of RNAinverse optimizing partition function. 
\item NUPACK~\citep{Zadeh+:2010}, which utilizes decomposition to optimize ensemble defect.
\item NEMO~\citep{portela:2018unexpectedly}, which combines Nested Monte Carlo Search with domain-specific knowledge for RNA design. It relies on human-crafted rules to achieve strong \MFE-based metrics. We also have a version of NEMO without hard rules.
\item m2dRNAs~\citep{rubio2018multiobjective}, which applies genetic algorithm to optimize multiple objectives..
\item SAMFEO~\citep{zhou+:2023samfeo} optimizes ensemble objectives yielding \MFE solutions as byproducts.
\end{enumerate}
These baselines are widely adopted in the RNA design literature~\citep{rubio2018multiobjective, anderson+:2016, portela:2018unexpectedly, garcia+:2013rnaifold}.
We do not include recent machine learning–based methods, which generally remain less competitive in this setting.
\vspace{-0.2cm}
\subsubsection{Method Variants}
We name our approach FastDesign, which has two variants: \ours and X\ours.
The key difference lies in the divide–conquer–combine stage: \ours performs additional refinement steps on the complete sequences generated at the root node, improving the ensemble-based objective. As shown in the following evaluations, \ours achieves state-of-the-art performance with strong efficiency, while X\ours provides an extremely fast RNA design solution that still delivers competitive design quality compared to baselines. Our methods are implemented in C++ and Python, hardware and hyperparameters are in Supplementary Section~\ref{sec:expsets}.

\subsection{Native Structure Design}
\begin{figure*}[ht]
\begin{minipage}{0.53\textwidth}
\begin{threeparttable}
\captionof{table}{Comparison of RNA design methods on RNAsolo764.\tnote{a}}
\label{tab:rnasolo}

\setlength{\tabcolsep}{2.5pt}
\begin{tabular}{l|rr|rr|r}
\hline
 \multirow{2}{*}{Method}      
 & \multicolumn{2}{c|}{Best}  
 & \multicolumn{2}{c|}{\#Solved {\scriptsize (union)}}  
 & \multicolumn{1}{c}{avg.~time} \\
 & $p(\ystar | \vecx)^\uparrow$
 & $\text{NED}_\downarrow$
 & MFE$^\uparrow$\!\!
 & uMFE$^\uparrow$\!\! 
 & (secs)$_\downarrow$ \\
\hline
RNAinverse        & 0.193 & 0.185 & 449 & 426 & \textbf{7.1} \\ 
RNAinverse-pf     & 0.717 & 0.016 & 599 & 593 & 34.1 \\ 
NUPACK            & 0.535 & 0.027 & 509 & 493 & 20.7 \\ 
NEMO              & 0.527 & 0.035 & \underline{609} & \underline{605} & 440.5 \\ 
m2dRNAs           & 0.588 & 0.037 & 601 & 600 & 226.7 \\ 
SAMFEO            & \underline{0.729} & \textbf{0.013} & 608 & 602 & 208.2 \\ 
\hline
X\ours         & 0.717 & 0.017 & 607 & 603 & \underline{13.3} \\ 
\ours         & \textbf{0.732} & \textbf{0.013} & \textbf{611} & \textbf{606} & 37.6 \\
\hline
\end{tabular}

\begin{tablenotes}
\item[a] \textbf{Bold}: the best; \underline{underline}: the second best.
\end{tablenotes}
\end{threeparttable}
\end{minipage}
\hspace{0.2cm}
\begin{minipage}{0.4\textwidth}
\vspace{-1.0cm}
\captionof{table}{Ablation Study of Rival Search on RNAsolo764}\label{tab:ab_rnasolo}
\begin{tabular}{l|c|rr|r}
\hline
Method &Ablation & MFE$^\uparrow$\!\! & \!\!uMFE$^\uparrow$\!\! & Time$_\downarrow$ \\
\hline
X\ours & None & 607 & 603 & 13.3 \\
X\ours & RivalSearch & 604 & 601 & 12.6 \\
\hline
\ours &  None & 611 & 606 & 37.6 \\
\ours & RivalSearch & 609 & 602 & 37.1 \\
\hline
\end{tabular}
\end{minipage}
\begin{minipage}{0.58\textwidth}
\begin{threeparttable}
\vspace{0.1cm}
\captionof{table}{Comparison of RNA design methods on Eterna100.}\label{tab:eterna}
\setlength{\tabcolsep}{3.5pt}
\begin{tabular}{l|rr|rr|r}
\hline
 \multirow{2}{*}{Method}      
 & \multicolumn{2}{c|}{Best}  
 & \multicolumn{2}{c|}{\#Solved {\scriptsize (union)}}  
 & \multicolumn{1}{c}{avg.~time} \\
 & $p(\ystar | \vecx)^\uparrow$
 & $\text{NED}_\downarrow$
 & MFE$^\uparrow$\!\! 
 & uMFE$^\uparrow$\!\! 
 & (secs)$_\downarrow$ \\
\hline
RNAinverse  & 0.064 & 0.322 & 32 & 29 &  \underline{207.7}\\
RNAinverse-pf  & 0.571 & 0.045 & 77 & 72 & 407.1 \\
NUPACK & 0.242 & 0.081 & 34 & 34 &  412.0 \\
NEMO & 0.271 & 0.098 & \underline{79} & \underline{77} & 937.5 \\
NEMO w/o rules & 0.162 & 0.147 & 76 & 75 & 917.9 \\
m2dRNAs & 0.377 & 0.100 & 72 & 70 & 716.1 \\
SAMFEO-origin\tnote{$^\ast$} & \underline{0.581} & 0.052 & 77 & 74 & 279.1 \\
SAMFEO-imprvd\tnote{$^\dagger$} & {0.580} & \underline{0.040} & \underline{79} & 75 & 673.8 \\
\hline
X\ours          & 0.537 & 0.060 &  76 &  74 & \textbf{127.5} \\
\ours          & \textbf{0.589} & \textbf{0.038} &  \textbf{80} &  \textbf{78} & 284.1 \\
\hline
\end{tabular}
\begin{tablenotes}
\item[$\ast$] SAMFEO-original, original  implementation~\cite{zhou+:2023samfeo}. 
\item[$\dagger$] \text{SAMFEO-improved, implementation with improved numerical accuracy.}
\end{tablenotes}
\end{threeparttable}
\end{minipage}
\hspace{0.1cm}
\begin{minipage}{0.4\textwidth}
\vspace{-2.0cm}
\captionof{table}{Ablation Study of Rival Search on Eterna100}\label{tab:ab_eterna}
\setlength{\tabcolsep}{2.5pt}
\begin{tabular}{l|c|rr|r}
\hline
Method &Ablation & MFE$^\uparrow$\!\! & \!\! uMFE$^\uparrow$\!\! & Time$_\downarrow$ \\
\hline
X\ours & None & 76 & 74 & 127.5 \\
X\ours & RivalSearch & 74 & 71 & 126.0 \\
\hline
\ours &  None & 80 & 78 &  284.1 \\
\ours & RivalSearch & 78 & 76 & 276.1 \\
\hline
\end{tabular}
    \includegraphics[width=1.05\linewidth]{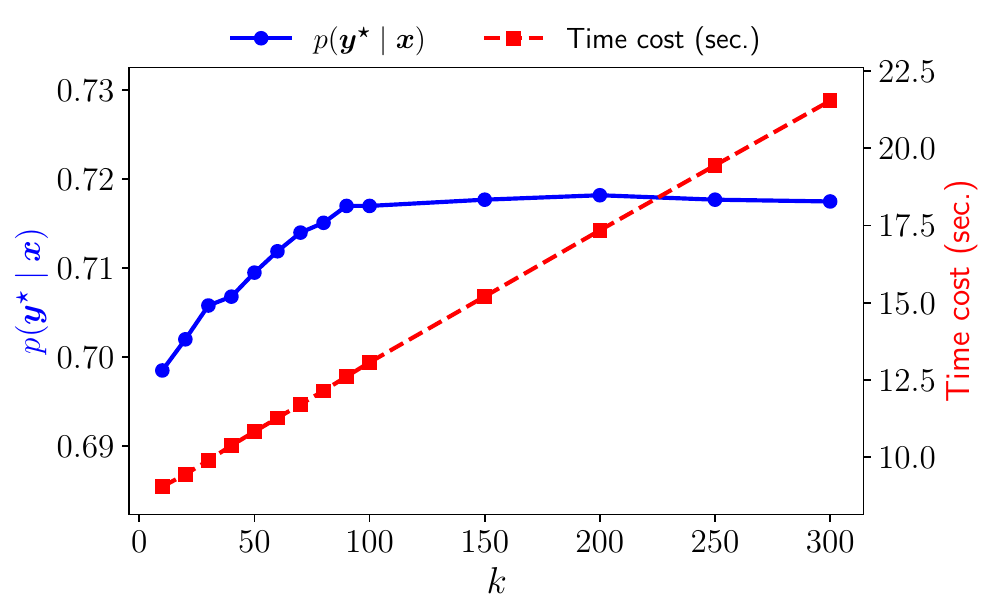}
    \vspace{-0.9cm}
    \captionof{figure}{Parameter of cube pruning size $k$ on RNAsolo.}\label{fig:param_k}
\end{minipage}
\end{figure*}
The results on RNAsolo764 are summarized in Table~\ref{tab:rnasolo}, and the results on ArchiveII are in Supplementary Section~\ref{sec:archiveii}.
\ours achieves the best performance across all metrics. For both $p(\ystar \mid \vecx)$ and \NED, it shows a clear improvement over all baseline methods except SAMFEO, which it still slightly surpasses.
In terms of \MFE and \UMFE success counts, \ours obtains the highest values (611 and 606). While NEMO also performs well on \MFE-based metrics, it falls behind significantly on ensemble-based measures.

Importantly, \ours is over 5× faster than SAMFEO and NEMO, highlighting the efficiency of our unified framework.
Moreover, X\ours achieves competitive design quality while being even faster—an order of magnitude quicker than most baselines. Given that RNAsolo structures are experimentally measured native structures, this suggests that in many biologically relevant scenarios, X\ours can provide high-quality results with extremely low computational cost. 
Note that NUPACK’s RNA folding parameters differ from those used in ViennaRNA. For a fair comparison, we implemented a special version of our method that uses NUPACK’s folding engine. The corresponding results are reported in Supplementary Section \ref{sec:results_nupack}.
\subsection{Artificial Structure Design}
Table~\ref{tab:eterna} reports results on Eterna100. The overall pattern is consistent with the RNAsolo results, demonstrating the robustness of our method across datasets.

\ours again achieves the best performance on every metric. Although SAMFEO produces high $p(\ystar \mid \vecx)$ values, \ours achieves an even higher overall score (0.589) while running more than twice as fast (284 vs. 673 seconds).
NEMO remains strong in \MFE-based success rates but continues to lag behind in ensemble-based metrics ($p(\ystar \mid \vecx)$ and \NED). The additional comparison with NEMO (w/o rules) further indicates that much of NEMO's strength comes from its hand-crafted domain rules.

On the other hand, X\ours is the fastest method among all competitors and still surpasses several baselines on various metrics.

\subsection{Ablation Study}\label{sec:ablation}
We ablated the component of rival structure search from our pipeline and report the results in ~\cref{tab:ab_rnasolo,tab:ab_eterna}. As we can see, the rival search consistently improved the number of solved structures under the \MFE and \UMFE metrics at a minor time cost. 
We also evaluated the effect of the cube pruning size $k$ within the divide–conquer–combine framework (Fig.~\ref{fig:param_k}). The design objective $p(\ystar \mid \vecx)$ improves as $k$ increases, saturating around $k=90$. The time cost increase grows nearly linearly with $k$.

\subsection{Long Structure Design }
\begin{table}[ht]
    \centering
    \captionof{table}{Comparison for long structure (16S-MFE) design.}\label{tab:16long}
    \renewcommand{\arraystretch}{1.1}
    \setlength{\tabcolsep}{2pt}
\resizebox{\linewidth}{!}{
    \begin{tabular}{c c cc rr}
    \hline
    \multirow{2}{*}{{ID}} & 
    \multirow{2}{*}{{Len}} & 
    \multicolumn{2}{c}{{Prob.~$p(\ystar \mid \vecx)$} $\uparrow$} &
    \multicolumn{2}{c}{{Time (secs)} $\downarrow$} \\
    \cline{3-6}
      & & {SAMFEO} & {X\ours} & {SAMFEO} & {X\ours} \\
    \hline
    1 & 950  & 0.267 & \textbf{0.365} & 17221.9 & \textbf{440.8} \\
    2 & 954  & 0.545 & \textbf{0.631} & 17121.9 & \textbf{640.4} \\
    3 & 1200 & 0.186 & \textbf{0.463} & 36174.7 & \textbf{2075.2} \\
    4 & 1452 & 0.017 & \textbf{0.117} & 65483.2 & \textbf{2544.6} \\
    5 & 1474 & 0.039 & \textbf{0.318} & 69006.3 & \textbf{9126.9} \\
    6 & 1474 & 0.203 & \textbf{0.507} & 75759.2 & \textbf{2645.4} \\
    7 & 1490 & 0.156 & \textbf{0.343} & 69896.0 & \textbf{3273.9} \\
    8 & 1492 & 0.216 & \textbf{0.518} & 72962.1 & \textbf{1301.3} \\
    9 & 1497 & 0.163 & \textbf{0.279} & 71942.7 & \textbf{1762.7} \\
    10 & 1525 & 0.062 & \textbf{0.396} & 71909.8 & \textbf{2317.2} \\
    \hline
{Avg.} & 1350.8 & 0.185 & \textbf{0.394} & 56747.8 & \textbf{2612.9} \\
    \hline
    \end{tabular}
}
\end{table}

Previous study~\citep{zhou+:2023samfeo} has shown that RNA design performance degrades sharply on the long 16SMFE structures, especially for folding probability $p(\ystar \mid \vecx)$.
For instance, SAMFEO achieves an average probability below 0.2, and other baselines produce near-zero probabilities. In addition, the cubic complexity of RNA folding causes their running times to grow rapidly with sequence length.

We therefore evaluate our method on the 16S-MFE benchmark, shown in Table~\ref{tab:16long}.
Compared to SAMFEO, X\ours doubles the folding probability, achieving an average of 0.394.
Furthermore, its runtime is an order of magnitude lower than SAMFEO.
These results further demonstrate the effectiveness and efficiency of our approach, particularly on challenging long-structure design tasks where existing methods struggle.

\vspace{-0.4cm}
\section{Related Work}

Many RNA design methods adopt local search, iteratively mutating and evaluating full sequences. Representative examples include RNAinverse~\citep{lorenz+:2011}, NEMO~\citep{portela:2018unexpectedly}, RNAiFold~\citep{garcia+:2013rnaifold}, and SAMFEO~\citep{zhou+:2023samfeo}.
Local mutations in these approaches are chosen through random sampling~\citep{lorenz+:2011, zhou+:2023samfeo}, manually crafted rules~\citep{portela:2018unexpectedly, taneda2011modena, rubio2018multiobjective}, or exhaustive enumeration~\citep{garcia+:2013rnaifold}.
Several machine-learning–based methods have also been proposed~\citep{runge2018learning, eastman2018solving, obonyo2022designing}, though their performance remains limited compared to state-of-the-art search-based algorithms.

Most prior work primarily targets the \MFE objective. NUPACK~\citep{Zadeh+:2010, wolfe2015sequence} emphasized the importance of ensemble-based metrics such as ensemble defect and introduced hierarchical structure decomposition for optimization. However, its decomposition largely occurs within long helices, and it redesigns leaf nodes rather than refining internal nodes, limiting flexibility. Other decomposition-based methods, such as RNA-SSD~\citep{andronescu2004new} and eM2dRNAs~\citep{rubio+:2023eM2dRNAs}, rely on heuristic rules for splitting structures and exhibit reduced design quality or computational efficiency.

To the best of our knowledge, our decomposition strategy is the first to be explicitly guided by motif-level designability, grounded in measurable RNA folding metrics (Eq.~\ref{eq:proof}). Moreover, our unified framework attains strong performance on both ensemble-based and \MFE-based objectives while maintaining high computational efficiency. These theoretical and empirical advantages set our approach apart from existing methods.

\vspace{-0.4cm}
\section{Conclusion and Future Work}

We propose a unified RNA design framework that achieves strong performance on both ensemble-based and \MFE-based objectives. At the motif level, our divide–conquer–combine strategy leverages designability theory together with an efficient cube pruning scheme to optimize folding probabilities. At the structure level, our rival-search procedure is rigorously derived from the necessary conditions of \MFE, enabling robust improvements in \MFE-based metrics.
Looking forward, several directions may further enhance our framework:
\begin{enumerate}
\item Improve the quality and diversity of rival structures to better guide the search for \MFE \& \UMFE designs.
\item Integrate  \ours-Fast and \ours-Full to achieve an improved balance between speed and design performance.
\end{enumerate}
\vspace{-0.8cm}

\bibliographystyle{natbib}
\bibliography{references}

\clearpage
\setcounter{section}{0}
\pagenumbering{roman}
\setcounter{figure}{0}
\renewcommand{\thefigure}{S\arabic{figure}}
\setcounter{table}{0}
\renewcommand{\thetable}{S\arabic{table}}
\renewcommand{\thesection}{S\arabic{section}}
\begin{center}
{\LARGE\bfseries Supplementary Information}
\end{center}

\section{Structural Loops}\label{sec:loops}

\begin{figure}[htbp]
	\centering
    \includegraphics[width=0.5\textwidth]{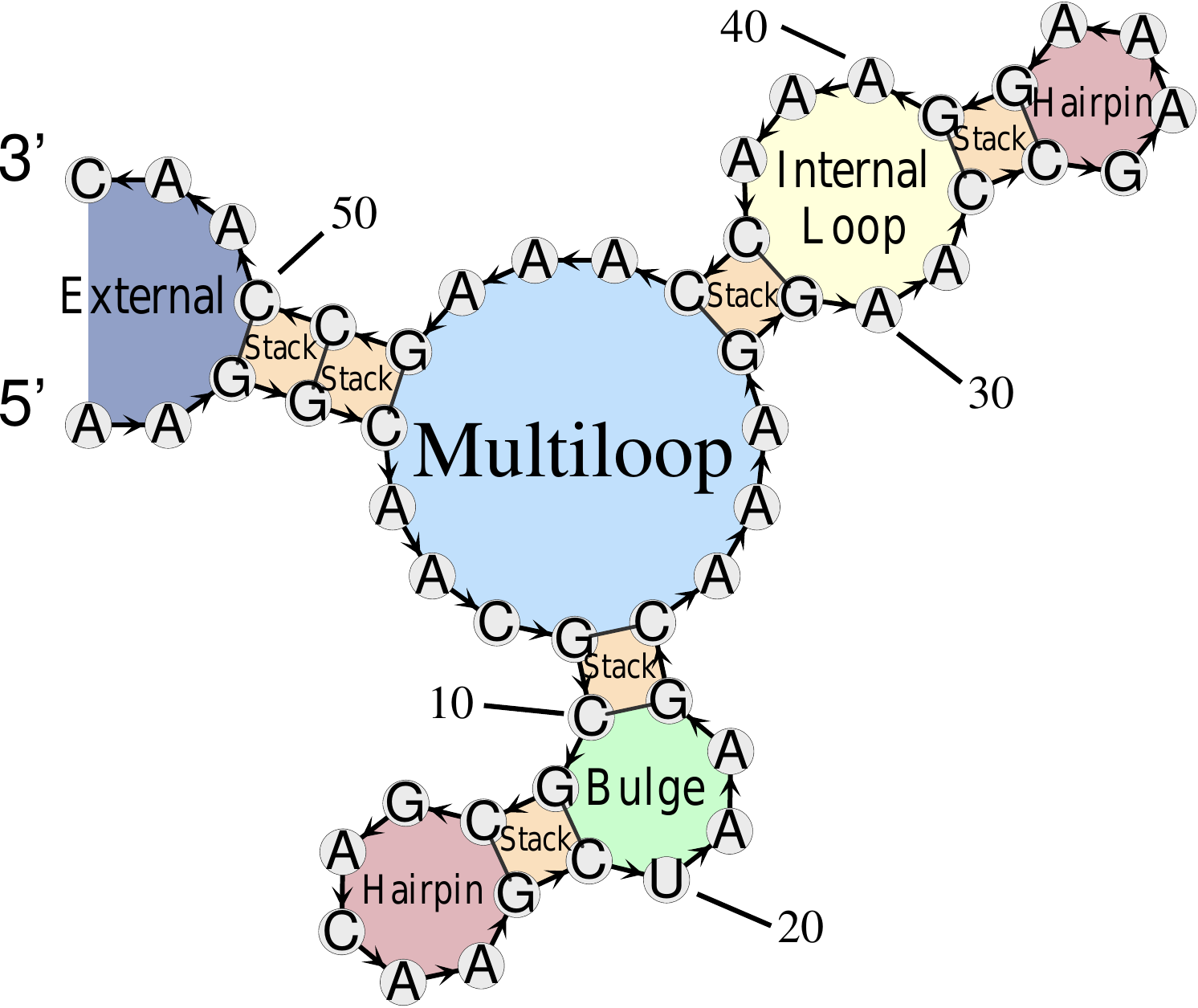}
    \captionof{figure}{An example of secondary structure and loops.} 
    \label{fig:loop}
\end{figure}

A secondary structure can be decomposed into a collection of loops, where each loop is usually a region enclosed by one or more base pairs. Depending on the number of pairs on the boundary, the main types of loops include hairpin loop, internal loop and multiloop, which are bounded by 1, 2, and 3 or more base pairs, respectively. In particular, the external loop is the most outside and is bounded by two ends ($5'$ and $3'$) and other base pair(s). Thus, each loop can be identified by a set of pairs. Fig.~\ref{fig:loop}  showcases an example of secondary structure with various types of loops, where some of the loops are notated as  

\begin{enumerate}
\item Hairpin: $H\blangle (12, 18) \brangle$.
\item Bulge: $B\blangle (10, 23), (11, 19) \brangle$.
\item Stack: $S\blangle (3, 50), (4, 49) \brangle$.
\item Internal Loop: $I\blangle (29, 43), (32, 39)\brangle$.
\item Multiloop: $M\blangle (5, 48), (9, 24), (28, 44)\brangle$.
\item External Loop: $E\blangle(3, 50) \brangle$.
\end{enumerate}

\section{Design Space/Constraint Intersection}\label{sec:intersection}

A (reduced) design space can be represented as a design constraint, which is defined as a set of indices $I$ and the specific nucleotides $\hat{\vecx}$ at these indices that must be satisfied. For convenience, we use $\hat{\vecx} = \vecx \proj I$ to denote the nucleotides
of a complete RNA sequence $\vecx$ that are ``projected" onto the indices $I$, which is described in Algorithm~\ref{alg:proj}. The intersection of two design constraints is described in Algorithm~\ref{alg:intersection}

\begin{table}[H]
\captionsetup{justification=centering}
\caption{Example of design constraint intersection: $C'_1, C'_2 = \contract(C_1, C_2)$\\ The composition for indicies (positions) $28, 29$ can only be $\nucG\nucC$ or $\nucG\nucU$}\label{tab:cc}
\makebox[\linewidth][c]{%
    \parbox[t]{.5\linewidth}{
        \centering
        \caption*{(a) Constraint $C_1$}
        \begin{tabular}{c|cccc}
            \hline
            	$I$ & 28 & 29 & 30 & 31 \\ \hline
		$\hat{\vecx}^1$ &G  & C  & C  & C  \\ \hline
		$\hat{\vecx}^2$ &G  & U  & C  & U  \\ \hline
		$\hat{\vecx}^3$ &G  & G  & U  & U  \\ \hline
		$\hat{\vecx}^4$ &G  & G  & U  & C  \\ \hline
        \end{tabular}
    }
}
    \\
 \makebox[\linewidth][c]{%
    \parbox[t]{.5\linewidth}{
        \centering
        \caption*{(b) Constraint $C_2$}
        
       \begin{tabular}{c|cccc}
            \hline
            	$I$ & 28 & 29 & 32 & 51 \\ \hline
		$\hat{\vecx}^1$ &G  & C  & A  & G  \\ \hline
		$\hat{\vecx}^2$ &G  & C  & A  & U  \\ \hline
		$\hat{\vecx}^3$ &G  & U  & U  & C  \\ \hline
		$\hat{\vecx}^4$ &A  & G  & U  & U  \\ \hline
        \end{tabular}
     }
 }
     \\
\makebox[\linewidth][c]{%
    \parbox[t]{.5\linewidth}{
        \centering
        \caption*{(c) Constraint $C'_1$}
        
       \begin{tabular}{c|cccc}
            \hline
            	$I$ & 28 & 29 & 30 & 31 \\ \hline
		$\hat{\vecx}^1$ &G  & C  & C  & C  \\ \hline
		$\hat{\vecx}^2$ &G  & U  & C  & U  \\ \hline
        \end{tabular}
     }
 }
\\
\makebox[\linewidth][c]{%
    \parbox[t]{.5\linewidth}{
        \centering
        \caption*{(d) Constraint $C'_2$}
        
       \begin{tabular}{c|cccc}
            \hline
            	$I$ & 28 & 29 & 32 & 51 \\ \hline
		$\hat{\vecx}^1$ &G  & C  & A  & G  \\ \hline
		$\hat{\vecx}^2$ &G  & C  & A  & U  \\ \hline
		$\hat{\vecx}^3$ &G  & U  & U  & C  \\ \hline
        \end{tabular}
     }
 }
\end{table}

\begin{algorithm}[ht]
\caption{Projection $\hat{\vecx} = \vecx \proj I$}\label{alg:proj}
\begin{algorithmic}[1]
    \Function{Projection}{$\vecx, I$} \Comment{$I = [i_1, i_2, \ldots, i_n]$ is a list of critical positions}
    \State $\hat{\vecx} \gets \text{map}()$
    \For{$i$ in $I$}
        \State $\hat{\vecx}[i] \gets \vecx_i$ \Comment{Project the $i$-th nucleotide to index $i$}
    \EndFor
    \State \Return $\hat{\vecx}$
    \EndFunction
\end{algorithmic}
\end{algorithm}

\begin{algorithm}[ht]
\caption{Contraint Intersection $C' = \text{Intersection}(C_1, C_2)$}\label{alg:intersection}
\begin{algorithmic}[1]
    \Function{Intersection}{$C_1, C_2$} \Comment{$C_1, C_2$ are sets of constraints}
     \State $(I_1, \setxi_1) \gets C_1$ \Comment{$I$ contains critical positions and $\setxi$ is a set of nucleotides compositions}
     \State $(I_2, \setxi_2) \gets C_2$
     \State $I' \gets I_1 \cap I_2$

     \If{$I' = \emptyset$} \Comment{No overlapping positions; return original constraints}
        \State \Return $C_1, C_2$ 
     \EndIf

     \State $\setxi'_1 \gets \{ \hat{\vecx} \proj I' \mid \hat{\vecx} \in \setxi_1 \}$
     \State $\setxi'_2 \gets \{ \hat{\vecx} \proj I' \mid \hat{\vecx} \in \setxi_2 \}$

     \For{$\hat{\vecx} \in \setxi_1$} \Comment{Remove nucleotides compositions from $\setxi_1$ that is not in $\setxi_2$}
        \If{$\hat{\vecx} \proj I' \notin \setxi'_2$}
        $\setxi_1 \gets \setxi_1 \setminus \{\hat{\vecx}\}$
        \EndIf
     \EndFor

     \For{$\hat{\vecx} \in \setxi_2$} \Comment{Remove nucleotides compositions from $\setxi_2$ that is not in $\setxi_1$}
        \If{$\hat{\vecx} \proj I' \notin \setxi'_1$}
        $\setxi_2 \gets \setxi_2 \setminus \{\hat{\vecx}\}$
        \EndIf
     \EndFor

     \State $C'_1 \gets (I_1,~\setxi_1)$
     \State $C'_2 \gets (I_2,~\setxi_2)$
     \State \Return $C'_1 \cup C'_2$     \Comment{Return updated constraints}
     \EndFunction
\end{algorithmic}
\end{algorithm}

\section{Divide-Conquer-Combine Framework}

The unified divide-conquer-combine is described in Algorithm~\ref{alg:dcc}, which calls Supp.~\cref{alg:samfeo-motif,alg:samfeo}.

\begin{algorithm*}[t]
\caption{\textsc{Divide-Conquer-Combine} Framework}
\label{alg:dcc}
\begin{algorithmic}[1]
\Function{DivideConqerAndCombine}{$\ystar$}
    \State $N_\rt \gets \textsc{Decompose}(\ystar)$\Comment{structure-level: Alg.~\ref{alg:decomp}}
    \State $X \gets \textsc{RecursiveConquerCombine}(N_\rt)$\Comment{root}
    \State $X_\prob, X_\NED, X_\MFE, X_\UMFE \gets \Call{SAMFEO}{\ystar, X, k_2, \step_2}$ \hspace{-0.1cm}\Comment{\text{\footnotesize Refine root (\ystar) designs via Supp.~Alg.~\ref{alg:samfeo}}}
    \State \Return $\vecx_\prob, \vecx_\NED, X_\MFE, X_\UMFE$\Comment{high-prob, low-NED, \MFE and \UMFE designs}
\EndFunction
\Function{RecursiveConquerCombine}{$N$}
    \If{$N$ is a leaf node}
        \State $X_\prob \gets \textsc{SAMFEO-Motif}(N.\vecm, \varnothing, k, \step_1)$ \Comment{Get leaf designs $X_\prob $ via Supp.~Alg.~\ref{alg:samfeo-motif}}
        \State \Return $X_\prob$
    \Else
        \State $\{N_1, N_2, \ldots, N_C\} \gets \textsc{Children}(N)$; $\{\vecm_1, \vecm_2, \ldots, \vecm_C\} \gets \{N_1.\vecm, N_2.\vecm, \ldots, N_C.\vecm\}$\Comment{Motifs of nodes}
        \For{$j = 1$ to $C$}
            \State $X_j \gets \textsc{RecursiveConquerCombine}(N_j)$
            \Comment{$X_j = \{ \vecx_1^j, \ldots, \vecx_k^j \}$ are the designs for  child $N_j$}
        \EndFor
        \State $X \gets \textsc{CubePruning}(X_1,...,X_C, \vecm_1,...,\vecm_C, k)$
        \Comment{$X$ denotes the designs for the parent $N.\vecm$}
        \If{$\nexists \vecx \in X$ \text{~s.t.~} \DG(\vecx, N.\vecm) = \MFE(\vecx)} \Comment{no \MFE designs in $X$?}
        \State $X_\prob \gets \textsc{SAMFEO-Motif}(N.\vecm, X, k, \step_r)$\Comment{Redesign}
        \EndIf
        \State \Return $X$
    \EndIf
\EndFunction

\end{algorithmic}
\end{algorithm*}

\section{Cube Pruning Algorithm}
The cube pruning algorithm of combining multiple motifs are show in Algorithm~\ref{alg:cubic_pruning_multi}.
\begin{algorithm}[htbp]
\caption{\textsc{CubePruning}$(\{X_d\}_{d=1}^D, \{\vec{m}_d\}_{d=1}^D, k)$ \Comment{version of combining multiple motifs}}
\label{alg:cubic_pruning_multi}
\begin{algorithmic}[1]

\Function{CubePruningMulti}{$\{X_d\}_{d=1}^D, \{\vec{m}_d\}_{d=1}^D, k$}
    \State $\vec{m}_c \gets \sum_{d=1}^D \vec{m}_d$ \Comment{Combine all motifs}
    \State Initialize a priority queue $\mathcal{Q}$
    \State Initialize an empty visited set $\mathcal{V}$
    \State Initialize result set $X_c \gets \varnothing$
    \State $\mathbf{i} \gets (1,1,\dots,1)$ \Comment{Multi-index for top-left cell}

    \State Compute $s = -\log p(\vec{m}_c \mid \sum_{d=1}^D \vec{x}^{(d)}_{1})$ by RNA folding
    \State Insert $(\mathbf{i}, s)$ into $\mathcal{Q}$ and mark $\mathbf{i}$ as visited

    \While{$|X_c| < k$ and $\mathcal{Q} \neq \varnothing$}

        \State $(\mathbf{i}, s) \gets$ \Call{PopBest}{$\mathcal{Q}$}
        \State Append the combined sequence $\sum_{d=1}^D \vec{x}^{(d)}_{i_d}$ to $X_c$ with score $s$

        \For{$d = 1$ to $D$}
            \State $\mathbf{i}' \gets \mathbf{i}$; increment $\mathbf{i}'[d] \gets \mathbf{i}[d] + 1$
            \If{$\mathbf{i}'$ within bounds and $\mathbf{i}' \notin \mathcal{V}$}
                \State Compute $s' = -\log p(\vec{m}_c \mid \sum_{d=1}^D \vec{x}^{(d)}_{i'_d})$
                \State Insert $(\mathbf{i}', s')$ into $\mathcal{Q}$
                \State Add $\mathbf{i}'$ to $\mathcal{V}$
            \EndIf
        \EndFor

    \EndWhile

    \State Sort $X_c$ by true score and return the top-$k$ results
    \State \Return $X_c$

\EndFunction
\end{algorithmic}
\end{algorithm}

\section{Rival Structure Search Algorithm}
The rival search algorithm is described in Algorithm~\ref{alg:rivalsearch}.
\begin{algorithm}[ht]
\caption{Rival Structure Search Algorithm (high-level version)}\label{alg:rivalsearch}
\begin{algorithmic}[1]
  \Statex $\mathcal{X}(\ystar < \vecy') = \{ \vecx \mid \DDG(\vecx, \ystar, \vecy') < 0\}$\Comment{design space: excluding sequences impossible for successful design}
   \Function{RivalStructureSearch}{$\ystar, \vecy$}\Comment{motif \ystar in a structure \vecy}
    \State $\yrival \gets \emptyset$\Comment{define a set of rival motifs}
    \State $X_\MFE \gets \emptyset$\Comment{a set of \MFE solutions}
    \State $X_\UMFE \gets \emptyset$\Comment{a set of \UMFE solutions}
    \While{$\bigcap_{\vecy' \in \yrival} \mathcal{X}(\ystar < \vecy') \neq \emptyset$ and $|\yrival| \leq N_r$} \Comment{design space is not empty; $N_r$: max. number of rivals }
	\For{$i = 1$ \textbf{to} $K$}
        		\State Draw $\vecx \in \bigcap_{\vecy' \in \yrival} \mathcal{X}(\ystar < \vecy')$
		\vspace{0.1cm}
		\If{$\DG(\vecx, \ystar) = \MFE(\vecx)$}
			\State $X_\MFE \gets X_\MFE \cup \{\vecx\}$
		\EndIf
		\If{$\forall \vecy \in \mathcal{Y}(\vecx), \DG(\vecx, \ystar) < \DG(\vecx, \vecy)$}
			\State $X_\MFE \gets X_\MFE \cup \{\vecx\}$
		\EndIf
		\If{$4^{|\D(\vecy', \ystar)|} < M_r$} 
        			 $\yrival \gets \yrival \cup \{\vecy''\}$ \Comment{limit the size of differential positions, $4^{|\D(\vecy', \ystar)|}$ denotes the required number of enumeration according to the differential positions between $\ystar$ and $\vecy'$~\cite{zhou+:2024undesignable}} 
		\EndIf 
        \EndFor
        \If{$\bigcap_{\vecy' \in \yrival} \mathcal{X}(\ystar < \vecy') = \emptyset$}
        	\vspace{0.1cm}
        		\State \Return $\ystar$ is {\color{blue} undesignable}
        \EndIf
    \EndWhile
    \State \Return $X_\MFE, X_\UMFE$    
    \EndFunction
\end{algorithmic}
\end{algorithm}

\section{Datasets Details}\label{sec:datasets}
Our experiments are conducted on a diverse set of RNA design benchmarks, including native structures, human-crafted structures, and structures derived from native RNA sequences.
\begin{enumerate}
\item \textbf{RNAsolo764.} RNAsolo~\citep{badura+:2025} is a comprehensive collection of RNA secondary structures curated from experimentally determined 3D structures~\citep{adamczyk+:2022rnasolo}, providing a challenging and biologically relevant benchmark for RNA design. To ensure quality and consistency, we deduplicated the structures and removed those containing prohibited loop types according to the standard folding rules of ViennaRNA. We further retained structures with a length of $\leq 400$~\emph{nt}. 
After preprocessing, 764 native RNA structures remained for our experiments. This dataset is available at \url{https://github.com/shanry/FastDesign/blob/main/data/rnasolo764.txt}.
%
%
\item \textbf{ArchiveII100}. ArchiveII~\citep{sloma+mathews:2016,Cannone+:2002} contains native RNA secondary structures spanning 10 families of naturally occurring RNAs, including tRNA, mRNA, rRNA, and others. It has been widely used for evaluating RNA folding methods~\citep{huang+:2019linearfold}. 
From the full dataset, we first removed all structures containing pseudoknots, as well as structures whose internal or multiloop configurations are prohibited by ViennaRNA. We also excluded structures that were previously proven to be undesignable~\cite{zhou+:2025scalable}. From the remaining pool, we selected 100 native structures such that no two structures have an edit distance less than or equal to 10, ensuring diversity within the benchmark.
This dataset is available at \url{https://github.com/shanry/FastDesign/blob/main/data/archiveii100.txt}.
\item \textbf{Eterna100.} The Eterna100 benchmark~\citep{anderson+:2016} consists of 100 secondary structures designed by human players of the online RNA puzzle game Eterna. Although it has become the most widely used benchmark in RNA design research recently, we note that Eterna100 structures often differ significantly from naturally occurring RNAs. Many are inspired by real-world shapes—such as “Chicken Feet” and “Gladius”—which enhances diversity but limits biological relevance.
\item \textbf{16S-MFE.} The 16S-MFE dataset~\citep{zhou+:2023samfeo} was introduced to assess the ability of RNA design methods to handle long RNA structures, since benchmarks like Eterna100 and Rfam generally include sequences shorter than 400~nt.
It contains structures folded from10 long sequences selected from 16S ribosomal RNAs in the ArchiveII database~\citep{sloma+mathews:2016,Cannone+:2002}. Each sequence was folded using the ViennaRNA 2.0 package, and the resulting minimum free energy (MFE) structures were used as design targets. This process ensures that the target structures are designable under MFE-based criteria while preserving biological relevance.
\end{enumerate}

\section{Experiment Settings}\label{sec:expsets}
Our algorithms are implemented in Python 3.11 (for the divide–conquer–combine stage) and C++ with g++ 11.5 (for the rival structure search).
All experiments, including baselines, were conducted on Linux (Rocky Linux 9.6) equipped with an AMD EPYC 9474F processor and 16~GB of memory.

We use unified hyperparameters for all variants of our method unless otherwise specified. The hyperparameters are as follows:

\begin{enumerate}
\item \textbf{Divide–conquer–combine stage.}
The number of steps for motif-level design is set to $\step_1 = 5000$ (Algorithm~\ref{alg:dcc}).
For internal node redesign which is invoked only when the combination condition is not met, we use $\step_r = 500$.
For root-level refinement, \ours uses $\step_2 = 2500$, while X\ours omits this stage with $\step_2 = 0$.
The cube pruning size is set to $k = 90$, as validated in the ablation study (Section~\ref{sec:ablation}).

\item \textbf{Rival structure search stage.}
We adopt the settings established in prior undesignability-identification work.
We set the maximum number of rival structures to $N_r = 100$, the maximum enumeration count to $M_r = 10^8$, and the maximum number of samples to $K = 500$ (Algorithm~\ref{alg:rivalsearch}).

\item \textbf{Parameters for SAMFEO and SAMFEO-MOTIF.}
For these baselines, we use the default parameters reported in the original SAMFEO work~\cite{zhou+:2023samfeo}.
\end{enumerate}

Most hyperparameters are fixed across all datasets. The primary tunable parameters are the number of steps $\step_1$ (SAMFEO-MOTIF), $\step_2$ (SAMFEO), and the cube pruning size $k$, which jointly control the trade-off between design quality and runtime.

By default, we decompose structures at the general easy-to-design motifs described in Section~\ref{sec:decompose}.
For long-structure design on 16S-MFE, we decompose only at helices due to their abundance in these large structures.

\section{Results on ArchiveII}\label{sec:archiveii}
\begin{table*}
\centering
\begin{threeparttable}
\captionof{table}{Results of RNA design methods on ArchiveII100.\tnote{a}}
\label{tab:archiveii}
\begin{tabular}{l|r|r|r|r|r}
\hline
 \multirow{2}{*}{Method}      
 & \multicolumn{2}{c|}{Best}  
 & \multicolumn{2}{c|}{\#Solved by Union}  
 & \multicolumn{1}{c}{Average} \\
 & $p(\ystar \mid \vecx)\uparrow$
 & $\text{NED}\downarrow$
 & MFE$\uparrow$ 
 & uMFE$\uparrow$ 
 & Time (sec.)$\downarrow$ \\
\hline
RNAinverse        & 0.0152 & 0.284 & 15 & 13 & \underline{188.6} \\ 
RNAinverse-pf     & 0.572 & 0.016 & 84 & 82 & 1143.0 \\ 
NEMO              & 0.225 & 0.035 & \textbf{91} & \underline{87} & {354.8} \\ 
m2dRNAs           & 0.305 & 0.036 & 85 & 82 & 773.0 \\ 
SAMFEO            & \underline{0.628} & \textbf{0.007} & 86 & 85 & 543.9 \\ 
\hline
X\ours         & 0.585 & \underline{0.013}  & \textbf{91} & \textbf{91} & \textbf{140.2} \\ 
\ours         & \textbf{0.658} & \textbf{0.007} & \textbf{91} & \textbf{91} & {424.4} \\
\hline
\end{tabular}
\begin{tablenotes}
\item[a] Bold = best, underline = second best.
\end{tablenotes}
\end{threeparttable}
\end{table*}

\section{Results with NUPACK's RNA Folding}\label{sec:results_nupack}
\begin{table}[H]
\centering
\captionof{table}{Results of RNA design methods on RNAsolo764.}\label{tab:results_nupack}
\begin{tabular}{|l|r|r|r|r|}
\hline
  \multirow{2}{*}{Method}      &  \multicolumn{2}{c|}{Best}  & \multicolumn{2}{c|}{ \#Solved by Union}  \\
 & $p(\ystar \mid \vecx)\uparrow$  & $\text{NED}\downarrow$& MFE$~\uparrow$ & uMFE$~\uparrow$  \\
\hline
NUPACK            & 0.535 & 0.027 & 509 & 493 \\ 
\hline
\ours  & {0.655} & {0.020} &  {557} &  {551} \\
\hline
\end{tabular}

\end{table}

\section{SAMFEO and Adapted SAMFEO}\label{sec:samfeo}

As SAMFEO~\citep{zhou+:2023samfeo} was originally developed for full-structure design, several key modifications are introduced to enable motif-level design.
\begin{enumerate}
\item The \textbf{input} is changed from a full RNA structure $\ystar$ to an individual target motif $\mstar$.
\item The \textbf{output} is changed from a complete RNA sequence to a partial sequence of the same length as $\ystar$. With a little notation abuse, $\vecx$ can denote either a complete RNA sequence or a partial RNA sequence depending on the context.
\item \textbf{Constrained folding} is used instead of unconstrained folding, since motif boundaries correspond to either fixed base pairs or the 5' \& 3' ends. Consequently, ensemble properties such as $p(\ystar | \vecx)$ and $\NED(\vecx, \ystar)$ are changed to $p(\mstar | \vecx)$ and $\NED(\vecx, \mstar)$ evaluated over a \emph{motif ensemble} rather than the full structural ensemble. The same principle is applied to the \MFE and \UMFE evaluation.
\item The original SAMFE has two variants:  $p(\ystar | \vecx)$ or $\NED(\vecx, \ystar)$ as the objective (fitness) function while \MFE \& \UMFE as byproducts. For adapted SAMFEO, we use $p(\ystar | \vecx)$ as the objective (fitness) function and produce \MFE \& \UMFE \& low-NED designs as byproducts.
\end{enumerate}

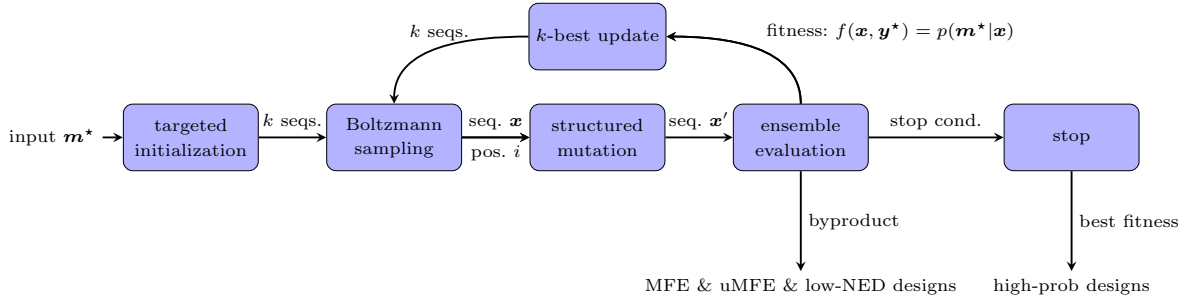
\begin{figure*}[t]
\centering
\vspace{-0.0cm}
\resizebox{0.99\textwidth}{!}{
	\begin{tikzpicture}[node distance=1cm]
	\node (input) [input] {input $\mstar$};
	\node (target_init) [process, right=of input, xshift=-0.7cm, align=center] {targeted \\ initialization};
	\node (boltzmann) [process, right=of target_init, align=center] {Boltzmann \\ sampling};
	\node (mutation) [process, right=of boltzmann, align=center] {structured \\ mutation};
	\node (k_best) [process, above=of mutation, yshift=-0.5cm] {$k$-best update};
	\node (ensemble) [process, right=of mutation, yshift=0cm, align=center] {ensemble \\ evaluation};
	\node (stop) [process, right=of ensemble, xshift=1cm, align=center] {stop};
	\node (output) [output, below=of ensemble, yshift=-0.5cm] {\MFE \& \UMFE \& low-\NED designs};
	\node (fbest)  [output, below=of stop, yshift=-0.5cm] {high-prob designs};

	\draw [arrow] (input) -- (target_init);
	\draw [arrow] (target_init) -- node[above] {$k$ seqs.} (boltzmann);
	\draw [arrow] (boltzmann) -- node[above] {seq. $\vecx$} (mutation);
	\draw [arrow] (boltzmann) -- node[below] {pos. $i$} (mutation);
	\draw [arrow] (mutation) -- node[above] {seq. $\vecx'$} (ensemble);
	\draw [arrow] (ensemble) -- node[above] {stop cond.} (stop);
	\draw [arrow] (ensemble) -- (output) node[midway,right] {byproduct};
	\draw [arrow] (stop) -- (fbest) node[midway,right] {best fitness};
	\draw [arrow] (ensemble) to[out=90, in=0] node[above] {\hspace{4cm}fitness: $f(\vecx, \ystar)=p(\mstar | \vecx)$} (k_best);
	\draw [arrow] (k_best) to[out=180, in=90] node[above] {$k$ seqs.} (boltzmann);
	\draw [arrow] (ensemble) to[out=90, in=0] (k_best);
	
	\end{tikzpicture}
}
\caption{Overview of adapted SAFEMO for motif design}\label{fig:motif_samfeo}
\end{figure*}

The workflow of the adapted SAMFEO is illustrated in Fig.~\ref{fig:motif_samfeo}. It follows the same iterative optimization cycle as the original algorithm, alternating between Boltzmann sampling, structured mutation, ensemble evaluation, and $k$-best update. However, all evaluations are restricted to the local motif ensemble rather than the global structure ensemble.
For detailed descriptions of the underlying optimization framework and scoring functions, we refer the reader to the original SAMFEO publication~\citep{zhou+:2023samfeo}.
The algorithms for SAMFEO and adapted SAMFEO are described in Algorithm~\ref{alg:samfeo} and Algorithm~\ref{alg:samfeo-motif} respectively.

\begin{algorithm}[t]
\caption{\textsc{SAMFEO}}
\label{alg:samfeo}
\begin{algorithmic}[1]
\Function{SAMFEO}{$\ystar$, $X_\text{init}$, $k_2$, $\step_2$}
    \State $X_{\MFE} \gets \varnothing$
    \State $X_{\UMFE} \gets \varnothing$
    \State $\vecx_{\NED} \gets \varnothing$
    \If{$X_\text{init}$ is $\varnothing$}
    	\State $X_\text{init} \gets \textsc{TargetedInitialization}(\ystar, k_2)$ \Comment{Initialize $k_2$ sequences via targeted initialization}
    \EndIf
    \State $X_{\prob} \gets X_\text{init}$
    \State $\vecx_{\prob} \gets \argmin_{\vecx \in X_\prob} p(\ystar \mid \vecx)$\Comment{The objective function is equilibrium probability}
    \For{$t = 1$ to $\step_2$}
        \State $\vecx \gets \textsc{SampleSequence}(X_{\prob})$  
        \State $i \gets \textsc{SamplePosition}(\vecx)$      
        \State $\vecx_{\new} \gets \textsc{StructuredMutation}(\vecx, i)$
        \State Evaluate $p(\ystar \mid \vecx_{\new})$ and $\NED(\vecx_{\new}, \ystar)$
        \If{$p(\ystar \mid \vecx_{\new}) > \min_{\vecx \in X_\prob} p(\ystar \mid \vecx)$}
        		\State  {Update} $X_{\prob}$ with $\vecx_{\new}$
        \EndIf
        \If{$p(\ystar \mid \vecx_{\new}) > p(\ystar \mid \vecx_{\prob})$}
        		\State $\vecx_{\prob} = \vecx_\new$
        \EndIf
        \If{$\NED(\vecx_{\new}, \ystar) < \NED(\vecx_{\NED}, \ystar)$}
        		\State $\vecx_{\NED} = \vecx_\new$
        \EndIf
        \If{$\ystar$ is an \MFE structure of $\vecx_\new$}
            \State $X_{\MFE} \gets X_{\MFE} \cup \{x_{\new}\}$
        \EndIf
        \If{$\ystar$ is the \UMFE structure of $\vecx_\new$}
            \State $X_{\UMFE} \gets X_{\UMFE} \cup \{x_{\new}\}$
        \EndIf
        \If{convergence condition is satisfied}
        	\State \textbf{break}
        \EndIf
    \EndFor
    
    \State \Return $\vecx_{\prob}, \vecx_\NED, X_{\MFE}, X_\UMFE$
\EndFunction
\end{algorithmic}
\end{algorithm}

\begin{algorithm}[t]
\caption{\textsc{SAMFEO-Motif}}
\label{alg:samfeo-motif}
\begin{algorithmic}[1]
\Function{SAMFEOMotif}{$\mstar$, $X_\text{init}$, $k_1$, $\step_2$}
    \State $X_{\MFE} \gets \varnothing$
    \State $X_{\UMFE} \gets \varnothing$
    \State $\vecx_{\NED} \gets \varnothing$
    \If{$X_\text{init}$ is $\varnothing$}
    	\State $X_\text{init} \gets \textsc{TargetedInitialization}(\mstar, k_1)$ \Comment{Initialize $k_1$ sequences via targeted initialization}
    \EndIf
    \State $X_{\prob} \gets X_\text{init}$
    \State $\vecx_{\prob} \gets \argmin_{\vecx \in X_\prob} p(\mstar \mid \vecx)$\Comment{The objective function is equilibrium probability}
    \For{$t = 1$ to $\step_2$}
        \State $\vecx \gets \textsc{SampleSequence}(X_{\prob})$  
        \State $i \gets \textsc{SamplePosition}(\vecx)$      
        \State $\vecx_{\new} \gets \textsc{StructuredMutation}(\vecx, i)$
        \State Evaluate $p(\mstar \mid \vecx_{\new})$ and $\NED(\vecx_{\new}, \mstar)$
        \If{$p(\mstar \mid \vecx_{\new}) > \min_{\vecx \in X_\prob} p(\mstar \mid \vecx)$}
        		\State  {Update} $X_{\prob}$ with $\vecx_{\new}$
        \EndIf
        \If{$p(\mstar \mid \vecx_{\new}) > p(\mstar \mid \vecx_{\prob})$}
        		\State $\vecx_{\prob} = \vecx_\new$
        \EndIf
        \If{$\NED(\vecx_{\new}, \mstar) < \NED(\vecx_{\NED}, \mstar)$}
        		\State $\vecx_{\NED} = \vecx_\new$
        \EndIf
        \If{$\mstar$ is an \MFE motif of $\vecx_\new$}
            \State $X_{\MFE} \gets X_{\MFE} \cup \{\vecx_{\new}\}$
        \EndIf
        \If{$\mstar$ is the \UMFE motif of $\vecx_\new$}
            \State $X_{\UMFE} \gets X_{\UMFE} \cup \{\vecx_{\new}\}$
        \EndIf
        \If{convergence condition is satisfied}
        	\State \textbf{break}
        \EndIf
    \EndFor
    
    \State \Return $\vecx_{\prob}, \vecx_\NED, X_{\MFE}, X_\UMFE$
\EndFunction
\end{algorithmic}
\end{algorithm}

\section{Structural Moitf}\label{sec:motif}

\begin{figure}[H]
    \centering
  \begin{subfigure}[b]{0.25\textwidth}
    \includegraphics[width=0.7\textwidth]{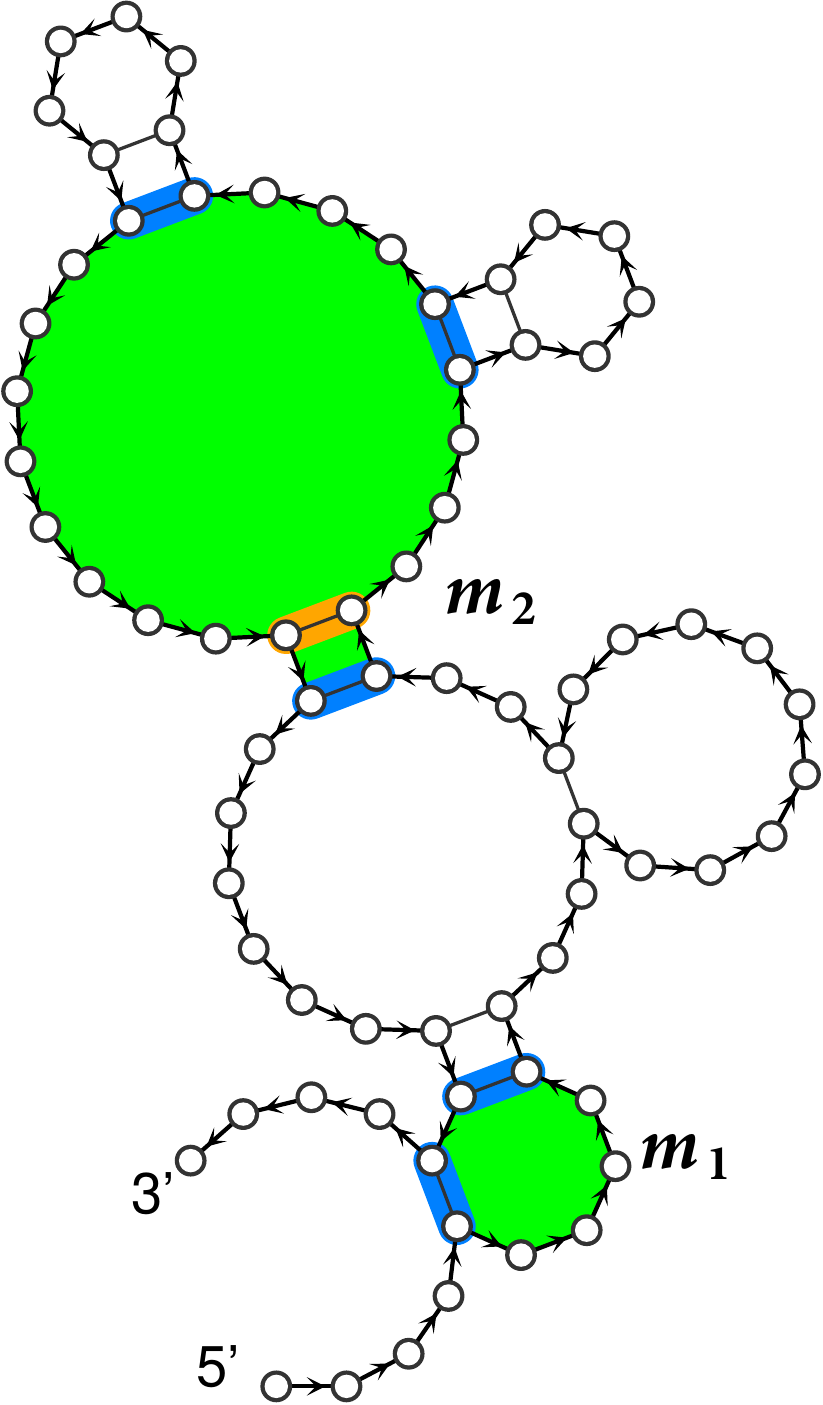}
  \end{subfigure}%
  \begin{subfigure}[b]{0.25\textwidth}
    \includegraphics[width=0.7\textwidth]{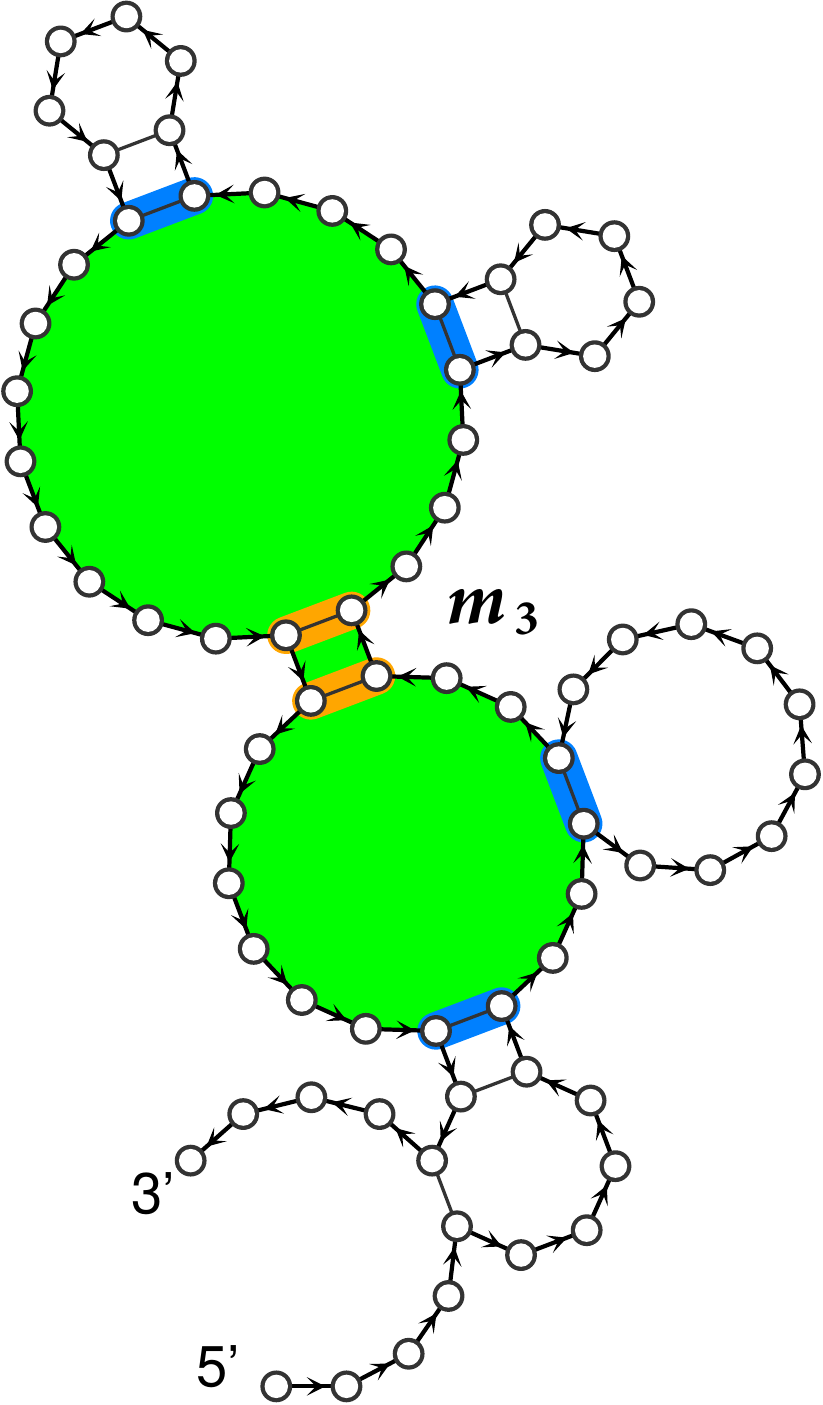}
  \end{subfigure}
  \caption{Motifs of various cardinalities (numbers of loops):
  $\card(\M_1)\!=\!1$, $\card(\M_2)\!=\!2$, $\card(\M_3)\!=\!3$.
  Loops are highlighted in green, internal pairs ($\ipairs$) in orange and {boundary pairs} ($\bpairs$) in blue.}\label{fig:card}
\end{figure}

\subsection{Motif is a Generalization of Structure}\label{subsec:motif}
\begin{definition}\label{def:motif-loops}
A \m \vecm is a contiguous (sub)set of loops in an RNA secondary structure \vecy, notated $\vecm\subseteq\vecy$.
 \end{definition}
 
 Many functions defined for secondary structures can also be applied to motifs. For example, $\LP(\M)$ represents the set of loops within a motif \M, while $\pairs(\M)$ and $\unpaired(\M)$ represent the sets of base pairs and unpaired positions, respectively. We define the \emph{cardinality} of \vecm  as the number of loops in \vecm , i.e., $\card(\vecm) = |\loops(\vecm)|$. Fig.~\ref{fig:card} illustrates three motifs, $\M_1, \M_2$, and $\M_3$, in a structure adapted from the Eterna puzzle ``\texttt{Cat's Toy}''. These motifs contain 1, 2, and 3 loops, respectively. We also define the \emph{length} of a motif $|\vecm|$ as the number of bases it contains, which is consistent with the length of a secondary structure $|\vecy|$.

Since motifs are defined as sets of loops, we can conveniently use set relations to describe their interactions. A motif $\M_A$ is a \emph{sub-motif} of another motif $\M_B$ if $\M_A$ is contained within $\M_B$, denoted as $\M_A \subseteq \M_B$. For the motifs in Fig.~\ref{fig:card}, we observe the relation $\M_2 \subseteq \M_3$. We further use $\M_A \subset \M_B$ to indicate that $\M_A$ is a proper sub-motif of $\M_B$, meaning $\M_A \neq \M_B$. Therefore, $\M_2 \subset \M_3$. The entire structure $\vecy$ can be regarded as the largest motif within itself, and accordingly, $\M \subseteq \vecy$ signifies that motif \M is a part of structure $\vecy$, with $\M \subset \vecy$ implying \M is strictly smaller than \vecy. 

The loops in a motif \M are connected by base pairs. Each base pair in $\pairs(\M)$ is classified as either an \emph{internal pair} linking two loops in \M or a \emph{boundary pair} connecting one loop inside \M to one outside. These two types of pairs in \M are denoted as disjoint sets $\ipairs(\M)$ and $\bpairs(\M)$, respectively:
\begin{align}
\ipairs(\M) \cap \bpairs(\M) = \emptyset,~\ipairs(\M) \cup \bpairs(\M) = \pairs(\M).
\end{align}
Utilizing the commonly accepted nearest neighbor model for RNA folding, it becomes evident that certain motifs may be absent from structures folded from RNA sequences. For instance, motif $\M_3$ in Fig.~\ref{fig:card} is considered undesignable, as the removal of its two internal pairs consistently reduces the free energy. This brings us to the definition of an  \emph{\um}.

\subsection{Motif Ensemble from Constrained Folding}
The designability of motifs is based on \emph{constrained folding}. Given a sequence $\vecx$, a structure in its ensemble $\vecy \in \mathcal{Y}(\vecx)$, we can conduct constrained folding by constraining the boundary pairs of $\M$, i.e., $\bpairs(\vecm)$. We generalize the concept of (structure) \emph{ensemble} to \emph{motif ensemble} as the set of motifs  that \vecx can possibly fold into (under the constraint $\bpairs(\vecm)$ being forced), denoted as $\mathcal{M}(\vecx, bpairs(\vecm))$. 
In the context of constrained folding, the folding outcomes are unaffected by the nucleotides at the constrained positions. Thus, with slight notation abuse, we use $\vecx$ to denote a partial sequence corresponding to a motif $\vecm$, where each position in $\vecx$ matches a position in $\vecm$, and vice versa. By this definition, the motif ensemble of a partial sequence $\vecx$ is denoted as $\mathcal{M}(\vecx)$. Similarly, the notation $\mathcal{X}(\vecm)$ generalizes $\mathcal{X}(\vecy)$ and represents all (partial) RNA sequences whose motif ensembles contain $\vecm$.
Motifs in $\mathcal{M}(\vecx, $\bpairs(\vecm)$)$ have the same boundary pairs, i.e., 
\begin{equation}
\forall \vecm', \vecm'' \in \mathcal{M}(\vecx), \bpairs(\vecm') = \bpairs(\vecm'') = \bpairs(\vecm). 
\end{equation}

The \emph{free energy change} of a motif $\vecm$ is the sum of the free energy of the loops in \vecm,

\begin{equation}\label{eq:emotif}
\DG(\vecx, \vecm) = \sum_{\vecz\in \loops(\vecm)}  \DG(\vecx, \vecz).
\end{equation} 

The definitions of \MFE and \UMFE can also be generalized to motifs via constrained folding.
\begin{definition}
A motif $\mstar \subseteq \vecy$ is an \MFE motif of folding $\vecx$ under constraint $\bpairs(\vecm)$ , i.e., $\MFE(\vecx, \bpairs(\vecm))$, if and only if
\begin{equation}
\forall \vecm \in \mathcal{M}(\vecx, \bpairs(\vecm))  \text{ and }\vecm \ne \mstar , \DG(\vecx, \mstar) \leq  \DG(\vecx, \vecm). \label{def:mfe-motif}
\end{equation}

\end{definition}
\begin{definition} \label{def:umfe-motif}
A motif $\mstar \subseteq \vecy$ is an \UMFE motif of folding $\vecx$ under constraint $\bpairs(\vecm)$ , i.e., $\UMFE(\vecx, \bpairs(\vecm))$, if and only if 
\begin{equation}
\forall \vecm \in \mathcal{M}(\vecx, \bpairs(\vecm))  \text{ and }\vecm \ne \mstar , \DG(\vecx, \mstar) <  \DG(\vecx, \vecm). \label{eq:umfe-motif} 
\end{equation}
\end{definition}

Similarly, the equilibrium probability of a sequence folding into the motif is defined as, 
\begin{equation}
p(\vecm \mid \vecx) = \frac{e^{-\DG(\vecx, \vecm)/RT}}{Q(\vecx)} = \frac{e^{-\DG (\vecx, \vecm)/RT}}{\sum_{\vecm' \in \mathcal{M}(\vecx, \bpairs(\vecm))}e^{-\DG(\vecx, \vecm')/RT}}. \label{eq:prob_m}
\end{equation}

\end{document}